\UseRawInputEncoding
\documentclass[journal]{IEEEtran}\IEEEoverridecommandlockouts  

\ifCLASSINFOpdf
\usepackage{graphicx}

\usepackage{amsthm}
\usepackage{algpseudocode} 
\usepackage{array}
\usepackage{amsmath,amsfonts} 
\usepackage{amssymb}
\usepackage{algorithm}
\usepackage{booktabs}
\usepackage{balance}
\usepackage[colorlinks,bookmarksopen,bookmarksnumbered,citecolor=black,urlcolor=black]{hyperref}
\usepackage{comment} 
\usepackage{caption}            
\usepackage{cite}  
\usepackage{epsfig}    
\usepackage{float}
\usepackage{fancyhdr}
\usepackage{graphicx}
\usepackage{indentfirst}
\usepackage{multirow}
\usepackage{subfigure}
\usepackage{stfloats}
\usepackage{subfigure}
\usepackage{verbatim}
\usepackage{url}
\usepackage{xcolor,graphicx}  

\newtheorem{theorem}{Theorem}

\renewcommand{\eqref}[1]{Equation (\ref{#1})}

\begin{document}

\title{Joint User Association, Interference Cancellation and Power Control for Multi-IRS Assisted UAV Communications}

\markboth{IEEE Transaction on Wireless Communications,~Vol.~1, No.~1}

\maketitle
\begin{abstract}
Intelligent reflecting surface (IRS)-assisted unmanned aerial vehicle (UAV) communications are expected to alleviate the load of ground base stations in a cost-effective way. Existing studies mainly focus on the deployment and resource allocation of a single IRS instead of multiple IRSs, whereas it is extremely challenging for joint multi-IRS multi-user association in UAV communications with constrained reflecting resources and dynamic scenarios. To address the aforementioned challenges, we propose a new optimization algorithm for joint IRS-user association, trajectory optimization of UAVs, successive interference cancellation (SIC) decoding order scheduling and power allocation to maximize system energy efficiency. We first propose an inverse soft-Q learning-based algorithm to optimize multi-IRS multi-user association. Then, SCA and Dinkelbach-based algorithm are leveraged to optimize UAV trajectory followed by the optimization of SIC decoding order scheduling and power allocation. Finally, theoretical analysis and performance results show significant advantages of the designed algorithm in convergence rate and energy efficiency.
\end{abstract}


\IEEEpeerreviewmaketitle
\IEEEpeerreviewmaketitle

\section{Introduction}
\IEEEPARstart{I}{n} recent years, the communications demands of wireless devices have surged, which poses a significant challenge to base stations (BSs) due to their overloaded traffic. This can also cause network congestion and reduce quality of service for users. Unmanned aerial vehicles (UAVs) play an important role in six-generation (6G) networks. Due to their flexibility, low costs and easy deployment, UAVs are widely used in wireless communication systems to build temporary networks for areas with weak network coverage or congestion. Compared to traditional wireless communications, UAVs as aerial BSs can realize better communication performance by staying close to users. 

\indent While UAVs bring the above benefits to wireless communication networks, they also bring new challenges. Since the wireless channels between UAVs and ground users are dynamic and unpredictable, the uncertainty of the signal propagation environment leads to the possibility of serious fading and Doppler shift. Especially in scenarios with high buildings, such as urban areas, the non-line of sight (NLoS) channel is easily encountered between the UAV and users, which may lead to signal attenuation and low transmission rate. 

\indent To solve the above problem, most of existing studies try to improve channel conditions by adjusting the horizontal position and height of UAVs. However, the lifetime of UAVs is shortened with frequent adjustments of their positions, due to the limitation of their battery capacities. Recent studies show that channel conditions between UAVs and ground users can be improved by deploying low-cost intelligent reflective surfaces (IRSs) on building surfaces. As a new technology of 6G, the IRS has wide applications in the field of wireless communications. It can be used to realize passive beamforming by phase shift adjustment and improve the achievable transmission rate of the UAV communication system. Therefore, the integration of IRSs and UAVs is promising for wireless channel enhancement. Most existing studies consider deploying a single IRS in the UAV communication system~\cite{9454446,10070838,10044705}. However, as the number of users increases, this may result in a large number of users unable to be served by the IRS. One promising approach is to deploy multiple IRSs and determine user association between users and IRSs, to improve system capacity and quality of service (QoS). Some researchers investigate communication enhancement with multi-IRS deployment, e.g.,~\cite{9893192,9870557,9804341}. However, in these studies, the service area of each IRS is assumed to be independent, e.g., each IRS only considers users within its own service area, and their covered areas are not overlapped~\cite{9870557,9804341}. Correspondingly, the user association remains independent, which means that the choices of IRSs do not affect others. In addition, all these studies consider UAVs hovering in fixed locations with static user association. 

\indent However, in the real-world scenario, the user association of a moving UAV is highly dynamic and the optimal user association needs to be selected in real-time based on the channel condition and the UAV position. In multi-IRS scenarios, most studies use one row of reflecting elements to serve one user according to the uniform linear array (ULA)~\cite{9893192}, which does not make full utilization of the communication resources of the system. Since IRSs have multiple rows and columns, uniform planar array (UPA) is now considered to obtain beamforming gain~\cite{9293155,9656117,9894720}. By introducing UPA, we can segment each IRS so that it can serve multiple users simultaneously. Although the authors in~\cite{9454446} consider the area segmentation of IRSs, it is an ideal static segmentation. Note that since the user association is dynamic, the segmentation is also variable with the trajectory of UAVs. In addition, much research on trajectory optimization for IRS-assisted UAV communications does not consider the energy consumption of UAVs~\cite{9866052,10.1109/TWC.2022.3212830,9749020}. While the authors in~\cite{9804220} optimize the flight and communication energy of UAVs, they do not consider the optimization of transmission rate.

\indent The flexibility of UAVs and early deployment of IRSs provide great flexibility to wireless communication networks. Despite the above research discussing IRS-assisted UAV communication systems, there are still some challenges that need to be addressed:

\begin{itemize}
	\item First, UAV communication systems are highly dynamic. However, existing studies usually consider single IRS deployment or simple IRS association scheduling, which cannot make full utilization of IRS resources. Therefore, it is necessary to design a dynamically segmentable multi-IRS user association scheduling scheme. 
	
     \item Second, the energy consumption of UAV flight and hovering has always been an issue that needs to be tackled. Existing research designs focus on how to maximize the system sum rate, ignoring the energy consumption of UAVs, which greatly affects system stability. Therefore, UAV trajectory optimization with energy efficiency as the optimization objective is worth studying.

	\item Finally, most of the current solutions employ traditional methods to solve the problems in IRS-assisted UAV communication systems, such as convex optimization and heuristic algorithms. They are inefficient in UAV communication systems with changing environments. Therefore, efficient learning algorithms are needed for user association in highly dynamic systems.

\end{itemize}

\indent To address the above challenges, we construct a non-orthogonal multiple access (NOMA)-based UAV communication model with multi-IRS deployment, and design an \underline{a}lternating optimization algorithm based on \underline{i}nverse \underline{s}oft-Q \underline{l}earning~\cite{garg2021iqlearn} and successiv\underline{e} convex approximation (SCA), named AISLE. The algorithm jointly optimizes user association, UAV trajectories, power allocation and successive interference cancellation (SIC) decoding orders to maximize system energy efficiency. We summarize the contributions of this paper as follows:

\begin{itemize}
	\item We construct a multiple IRS-assisted UAV communication system, and design an IRS dynamic segmentation strategy that allows a single IRS to serve multiple users simultaneously. To improve the spectrum utilization of the system, we leverage the NOMA communication technique and optimize the SIC decoding order.
	
     \item Considering the flight and communication energy consumption of the UAV, we formulate an optimization problem for maximizing the system energy efficiency. To solve this problem, we first decompose it into three subproblems based on the coupling of variables: IRS-user association, UAV trajectory design, and joint SIC decoding order scheduling and power allocation. An alternating optimization algorithm is proposed to solve these subproblems.

	\item For the first subproblem, we design an inverse soft-Q learning based algorithm to decide user association. For the second subproblem, we design a Dinkelbach and SCA-based algorithm to approximate the optimal UAV trajectory. Finally, we design a penalty-based SCA algorithm to optimize SIC decoding scheduling and power allocation. In addition, we theoretically prove the algorithm effectiveness from convergence and time complexity.


\end{itemize}

\indent The remainder of this paper is organized as follows: In Section II, we introduce the UAV communication system model with multi-IRS deployment, and formulate the system energy efficiency maximization problem. The AISLE algorithm is presented in Section III. Section IV presents the simulation results of the designed algorithm. Finally, we provide a brief conclusion in Section V. 

\indent \textit{Notation:} Vectors are denoted by bold-face, and $\left\| \cdot  \right\|$ represents the Euclidean norm of a vector. Symbol ${{\mathbb{C}}^{M\times N}}$ represents an $M\times N$ matrix. Symbols ${\mathbf{X}^\mathbf{T}}$ and ${\mathbf{X}^\mathbf{H}}$ denote the transpose and conjugate of $\mathbf{X}$, respectively, and $\mathbf{X}\otimes \mathbf{Y}$ denotes the Kronecker product of $\mathbf{X}$ and $\mathbf{Y}$. Operator $\left\lfloor \cdot \right\rfloor$ denotes the downward rounding, and $\mathbb{E}\left( \cdot  \right)$ represents the ematical expectation.

\begin{figure}
  \centering
  \includegraphics[width=0.80\linewidth]{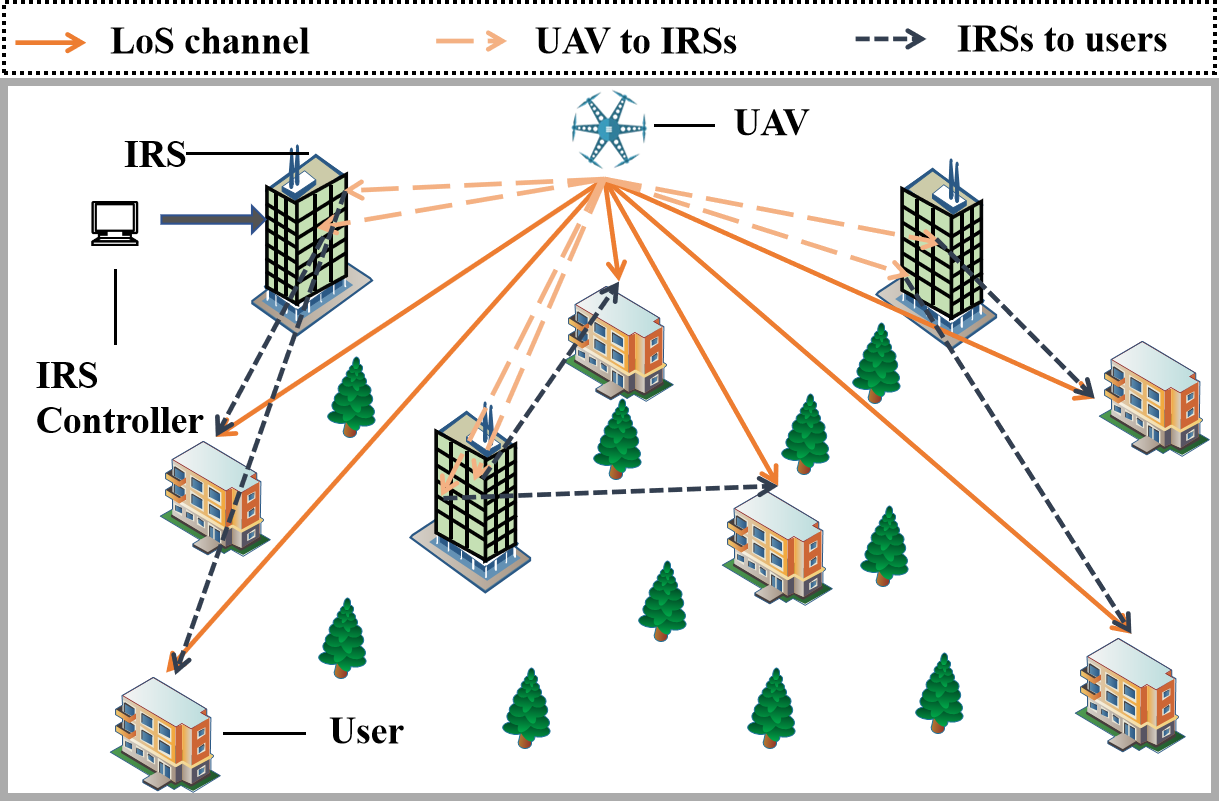}
  \caption{The system model of IRS-assisted UAV communications.}
  \label{fig1}
\vspace{-5mm}
\end{figure}
\section{System Model}

\indent Fig. 1 shows the considered multi-IRS assisted UAV communication system to support ground users' communication in traffic-congested areas. Specifically, a NOMA-based UAV communication system with multi-IRS deployment is considered to serve $N$ users, where set $\mathcal{N}=\left\{1,..., i,..., N \right\}$ denotes the index of users. Each IRS consists of ${{J}_{b}}\times {{J}_{I}}$ reflecting elements that are deployed on the surface of buildings to ensure that all users in the system can obtain their services. The index of IRSs is denoted by set $\mathcal {S}=\left\{ 1,2,...,s,...,S \right\}$. Symbol $\mathcal{T}$ denotes the total service time and is divided into $T$ equal time slots, i.e., $\mathcal{T}=T\tau$, where $\tau$ denotes time slot. In each time slot $\tau$, users communicate with the UAV by NOMA. Without loss of generality, three-dimensional (3D) Euclidean coordinates are used for modeling, and the UAV coordinate at time slot $t$ is denoted by $L[t]={{[x[t],y[t],z[t]]}^{\text{T}}}$. In addition, consider that the positions of both users and IRSs are fixed during the entire service time, and thus the coordinate of user $i$ is denoted by ${{L}_{i}}={{[{{x}_{i}},{{y}_{i}},{{z}_{i}}]}^{\text{T}}}$, and that of IRS $s$ is denoted by ${{L}_{s}}={{[{{x}_{s}},{{y}_{s}},{{z}_{s}}]}^{\text{T}}}$. The Euclidean distance between the UAV and user $i$ at time slot $t$ is denoted by $d_{i}^{\text{UG}}[t]$, and that between the UAV and IRS $s$ is represented by $d_{s}^{\text{UR}}[t]$, i.e., $d_{i}^{\text{UG}}={\left\| L[t]-{{L}_{i}} \right\|}$ and $d_{s}^{\text{UR}}=\left\| L[t]-{{L}_{s}} \right\|$. At time slot $t$, we assume that the UAV position is invariant, because the position change of the UAV is considerably smaller than distances $d_{i}^{\text{UG}}[t]$ and $d_{s}^{\text{UR}}[t]$~\cite{9293155,9804220}. In addition, considering the path loss and reflection loss in free space, we do not consider signals that are reflected twice and more~\cite{9804220}. To resemble a real congested communication system, we assume that each user in each time slot has a communication demand.

\begin{figure*}[b]
\hrulefill
	\begin{align*}\label{eq3}
\small
\textbf{h}_{s,i}^{\text{UR}}\left[ t \right] &= \left[ 1, e^{-j2\pi \frac{f_c}{c}D_I\sin \theta _{s}^{\text{UR}}\left[ t \right] \cos \varphi _{s}^{\text{UR}}\left[ t \right]}, \dots, e^{-j2\pi \frac{f_c}{c}\left( J_I-1 \right)D_I\sin \theta _{s}^{\text{UR}}\left[ t \right] \cos \varphi _{s}^{\text{UR}}\left[ t \right]} \right]^{\text{T}} \tag{3}\\ 
&\quad \otimes \left[ 1, e^{-j2\pi \frac{f_c}{c}D_b\sin \theta _{s}^{\text{UR}}\left[ t \right] \sin \varphi _{s}^{\text{UR}}\left[ t \right]}, \dots, e^{-j2\pi \frac{f_c}{c}\left( J_{s,i}\left[ t \right] -1 \right) D_b\sin \theta _{s}^{\text{UR}}\left[ t \right] \sin \varphi _{s}^{\text{UR}}\left[ t \right]} \right]^{\text{T}}.
	\end{align*}
\end{figure*}

\subsection{Channel Model}

\indent We consider a multi-IRS assisted UAV communication system, and the channel between the UAV and user $i$ consists of a direct channel and a reflected channel through IRSs. The direct channel is modeled according to the Rician fading channel model in the city~\cite{9454446} and can be expressed by: 
\begin{align*}\label{eq1}
\small
H_{i}^{\text{UG}}[t] &=\sqrt{\frac{\beta }{{{(d_{i}^{\text{UG}}[t])}^{\xi _{i}^{\text{UG}}}}}}\left( \sqrt{\frac{\gamma _{i}^{\text{UG}}}{\gamma _{i}^{\text{UG}}+1}}\right.\tag{1}\\ 
&\left. +\sqrt{\frac{1}{\gamma _{i}^{\text{UG}}+1}}\tilde{h}_{i}^{\text{UG}}\left[ t \right] \right),
\end{align*}
\noindent where $\beta$ denotes the path loss at one meter. Variable $\xi _{i}^{\text{UG}}$ denotes the path loss exponent corresponding to the channel between the UAV and user $i$. Symbol $\gamma _{i}^{\text{UG}}$ denotes the corresponding Rician factor, and $\tilde{h}_{i}^{\text{UG}}\left[ t \right]$ is the random NLoS scattering component obeying a cyclic symmetric Gaussian distribution. 

\indent The reflected channel comprises the channel from the UAV to IRS $s$ and the channel from IRS $s$ to user $i$, and both of them are considered as line of sight (LoS) channels. Thus, we model the channel between the UAV and IRS $s$ by employing the generic LoS channel model~\cite{9454446}, and the channel vector at time slot $t$ is denoted by:
\begin{align*}\label{eq2}
\small
\mathbf{H}_{s,i}^{\text{UR}}\left[ t \right]=\sqrt{\frac{\beta }{{{(d_{s}^{\text{UR}}[t])}^{2}}}}\mathbf{h}_{s,i}^{\text{UR}}[t].\tag{2}
\end{align*}
\indent Let $f_{c}$ and $c$ denote the carrier frequency and light speed, respectively, and LoS component $\mathbf{h}_{s,i}^{\text{UR}}[t]$ is denoted by equation (\ref{eq3}) at the bottom of this page.

\indent We design an efficient segmentation strategy to enable IRS $s$ to serve multiple users simultaneously. Inspired by UPA and~\cite{9454446,9976948,9714216}, we adopt rectangular arrays to divide IRSs. Specifically, since the rectangular array is used in the current studies~\cite{9784946,9810528,10109153}, the rectangular segmentation is more effective in utilizing the area of IRSs compared to the circular and elliptical segmentation. Let binary variable ${{\alpha }_{s,i}}[t]\in \left\{ 0,1 \right\}$ denotes whether user $i$ is served by IRS $s$ at time slot $t$, where ${{\alpha }_{s,i}}[t] = 1$ indicates user $i$ is served by IRS $s$ and vice versa. To facilitate the design, IRSs are divided evenly among the served users since user association variable ${{\alpha }_{s,i}}$ varies at time slot $t$. Therefore, we consider an equal-sized rectangular array to be assigned to the served users after the rows of IRSs are uniformly partitioned. In this paper and related work~\cite{9913496,9681874}, the coordinates of the reflecting elements are replaced by the coordinates of the reflective surfaces. Thus, the segmentation of rows and columns of IRSs is equivalent.

\indent Note that the limit case is transformed from UPA to ULA when the number of users served by IRSs matches the number of rows. Variable ${{J}_{s,i}}[t]$ denotes the row number of reflecting elements assigned to user $i$ by IRS $s$ at time slot $t$. By considering that the reflecting element is the smallest unit of IRS, expression ${{J}_{s,i}}[t]=\left\lfloor {{J}_{b}}/\sum\limits_{i=1}^{N}{{{\alpha }_{s,i}}\left[ t \right]} \right\rfloor$ is founded. Variables ${{D}_{b}}$ and ${{D}_{I}}$ denote the row spacing and column spacing of the corresponding reflecting elements, respectively. Variables $\theta _{s}^{\text{UR}}[t]$ and $\varphi _{s}^{\text{UR}}[t]$ denote the vertical and horizontal angles of arrival of IRS $s$, then $\sin \theta _{s}^{\text{UR}}[t]=\frac{z[t]-{{z}_{s}}}{d_{s}^{\text{UR}}[t]}$, $\sin \varphi _{s}^{\text{UR}}[t]=\frac{{{x}_{s}}-x[t]}{\sqrt{{{\left( {{x}_{s}}-x[t] \right)}^{2}}+{{\left( {{y}_{s}}-y[t] \right)}^{2}}}}$, $\cos \varphi _{s}^{\text{UR}}[t]=\frac{y[t]-{{y}_{s}}}{\sqrt{{{\left( {{x}_{s}}-x[t] \right)}^{2}}+{{\left( {{y}_{s}}-y[t] \right)}^{2}}}}$.

\indent Using the Rician fading model, the channel vector between user $i$ and IRS $s$ at time slot $t$ is denoted by:
\begin{align*}\label{eq4}
\small
\mathbf{H}_{s,i}^{\text{RG}}[t]&=\sqrt{\frac{\beta }{{{(d_{s,i}^{\text{RG}})}^{\xi _{s,i}^{\text{RG}}}}}}\left( \sqrt{\frac{\gamma _{s,i}^{\text{RG}}}{\gamma _{s,i}^{\text{RG}}+1}}\mathbf{h}_{s,i}^{\text{RG}}[t] \right.\tag{4}\\ 
&\left.+\sqrt{\frac{1}{\gamma _{s,i}^{\text{RG}}+1}}\tilde{\mathbf{h}}_{s,i}^{\text{RG}}\left[ t \right] \right),
\end{align*}
\noindent where $d_{s,i}^{\text{RG}}=\left\| {{L}_{s}}-{{L}_{i}} \right\|$ and $\xi _{s,i}^{\text{RG}}$ denotes the distance and path loss exponent of the channel between IRS $s$ and user $i$, respectively, and $\gamma _{s,i}^{\text{RG}}$ denotes the corresponding Rician factor. Variable $\tilde{h}_{s,i}^{\text{RG}}\left[ t \right]$ denotes the NLoS component, which is subject to a cyclic symmetric Gaussian distribution. The vertical and horizontal angles of arrival from IRS $s$ to user $i$ are denoted by $\theta _{s,i}^{\text{RG}}\left[ t \right]$ and $\varphi _{s,i}^{\text{RG}}\left[ t \right]$, respectively. Then, $\sin \theta _{s,i}^{\text{RG}}={z}_{s}/d_{s,i}^{\text{RG}}$, $\sin \varphi _{s,i}^{\text{RG}}[t]=\frac{{{x}_{i}}-{{x}_{s}}}{\sqrt{{{\left( {{x}_{s}}-{{x}_{i}} \right)}^{2}}+{{\left( {{y}_{s}}-{{y}_{i}} \right)}^{2}}}}$, and $\cos \varphi _{s,i}^{\text{RG}}[t]=\frac{{y}_{i}-{y}_{s}}{\sqrt{{{\left( {{x}_{s}}-{{x}_{i}} \right)}^{2}}+{{\left( {{y}_{s}}-{{y}_{i}} \right)}^{2}}}}$. LoS component $\mathbf{h}_{s,i}^{\text{RG}}\left[ t \right]$ is represented by equation (\ref{eq5}), located at the bottom of this page.
\begin{figure*}[b]
\hrulefill
	\begin{align*}\label{eq5}
\mathbf{h}_{s,i}^{\text{RG}}\left[ t \right] &=\left[ 1,e^{-j2\pi \frac{f_c}{c}D_I\sin \theta _{s,i}^{\text{RG}}\cos \varphi _{s,i}^{\text{RG}}},...,e^{-j2\pi \frac{f_c}{c}\left( J_I-1 \right)D_I\sin \theta _{s,i}^{\text{RG}}\cos \varphi _{s,i}^{\text{RG}}}\right] ^{\text{T}}\tag{5}\\ 
&\quad \otimes \left[ 1,e^{-j2\pi \frac{f_c}{c}D_b\sin \theta _{s,i}^{\text{RG}}\sin \varphi _{s,i}^{\text{RG}}},...,e^{-j2\pi \frac{f_c}{c}\left( J_{s,i}\left[ t \right] -1 \right) D_b\sin \theta _{s,i}^{\text{RG}}\sin \varphi _{s,i}^{\text{RG}}} \right] ^{\text{T}}.
	\end{align*}
\end{figure*}

\indent In particular, in order to fully utilize reflecting elements of IRSs, UPA is considered, and the corresponding reflection phase shift matrix is denoted by:
\begin{align*}\label{eq6}
\small
{{\Phi }_{s,i}}[t]=\text{diag}(\phi [t])\in {{\mathbb{C}}^{{{J}_{I}}{{J}_{s,i}}[t]\times {{J}_{I}}{{J}_{s,i}}[t]}}, \tag{6}
\end{align*}
\noindent where $\phi [t]={{\left[ {{e}^{j{{\phi }_{1,1}}[t]}},...,{{e}^{j{{\phi }_{{{J}_{I}},{{J}_{s.i}}[t]}}[t]}} \right]}^{\text{T}}}\in {{\mathbb{C}}^{{{J}_{I}}{{J}_{s,i}}[t]\times 1}}$. 

\indent The amplitude loss caused by the IRS is denoted by $A\in \left( 0,1 \right)$, and the reflected channel from the UAV to user $i$ via IRS $s$ is denoted by:
\begin{align*}\label{eq7}
\small
H_{s,i}^{\text{URG}}\left[ t \right]=A{{(\mathbf{H}_{s,i}^{\text{RG}}[t])}^{\mathbf{H}}}{{\Phi }_{s,i}}[t]\mathbf{H}_{s,i}^{\text{UR}}[t].\tag{7}
\end{align*}
\noindent Then, combining the directed channel and the reflected channel yields a composite channel from the UAV to user $i$, represented as:
\begin{align*}\label{eq8}
\tiny
O_{i}^{\text{UG}}[t]&=H_{i}^{\text{UG}}[t]+\sum\limits_{s=1}^{S}{{{\alpha }_{s,i}}[t]H_{s,i}^{\text{URG}}\left[ t \right]}\tag{8}\\
&=Q_{i}^{\text{UG}}[t]+\tilde{Q}_{i}^{\text{UG}}[t],
\end{align*}
\noindent where $Q_{i}^{\text{UG}}[t]$ and $\tilde{Q}_{i}^{\text{UG}}[t]$ denote the LoS component and the NLoS component, respectively. $Q_{i}^{\text{UG}}[t]$ is represented by equation (\ref{eq9}). Note that the LoS component dominates in UAV communication scenarios and changes slowly compared to the NLoS component~\cite{9293155}. Therefore, in the following, we mainly focus on the design of the LoS component. 

\begin{figure*}[b]
\hrulefill
	\begin{align*}\label{eq9}
Q_{i}^{\text{UG}}[t] =\sqrt{\frac{\beta }{{{(d_{i}^{\text{UG}}[t])}^{\xi _{i}^{\text{UG}}}}}}\sqrt{\frac{\gamma _{i}^{\text{UG}}}{\gamma _{i}^{\text{UG}}+1}} +\sum\limits_{s=1}^{S}{\frac{{{\alpha }_{s,i}}\left[ t \right]A\beta }{d_{s}^{\text{UR}}[t]{{(d_{s,i}^{\text{RG}})}^{\xi _{s,i}^{\text{RG}}/2}}}\sqrt{\frac{\gamma _{s,i}^{\text{RG}}}{\gamma _{s,i}^{\text{RG}}+1}}{{\left( \mathbf{h}_{s,i}^{\text{RG}}[t] \right)}^{\textbf{H}}}{{\Phi }_{s,i}}[t]\mathbf{h}_{s,i}^{\text{UR}}[t]}.\tag{9}
	\end{align*}
\end{figure*}

\indent Then, a phase shift control strategy is utilized and the strategy yields the optimal passive beamforming gain in IRS-assisted UAV communication scenarios~\cite{9866052}, and it's denoted by equation (\ref{eq10}) at the bottom of the next page.
\begin{figure*}[b]
\hrulefill
	\begin{align*}\label{eq10}
{{\phi }_{{{J}_{I}},{{J}_{s,i}}[t]}}[t]&=-2\pi \frac{{{f}_{c}}}{c}({{D}_{I}}({{J}_{I}}-1)\sin \theta _{s,i}^{\text{RG}}\cos \varphi _{s,i}^{\text{RG}}-{{D}_{b}}({{J}_{b}}-1)\sin \theta _{s,i}^{\text{RG}}\sin \varphi _{s,i}^{\text{RG}} \tag{10}\\ 
 & +{{D}_{I}}({{J}_{I}}-1)\sin \theta _{s}^{\text{UR}}[t]\cos \varphi _{s}^{\text{UR}}[t]+{{D}_{b}}({{J}_{b}}-1)\sin \theta _{s}^{\text{UR}}[t]\sin \varphi _{s}^{\text{UR}}[t]),
	\end{align*}
\end{figure*}
\begin{theorem}
	The phase control strategy denoted by equation (\ref{eq10}) is optimal at time slot $t$.
\end{theorem}
\begin{proof}
 \textit{The optimal phase control strategy is to obtain the maximum passive beamforming gain, i.e., maximize ${{\left| Q_{i}^{\text{UG}}[t] \right|}^{2}}$, by controlling the phase shift of IRS reflecting elements. Substituting equations (\ref{eq3}) and (\ref{eq5}) into equation (\ref{eq9}) yields:}
\begin{align*}\label{eq11}
\small
   Q_{i}^{\text{UG}}[t]=&{{\mu }_{i}}[t]\sqrt{\frac{\beta }{{{(d_{i}^{\text{UG}}[t])}^{\xi _{i}^{\text{UG}}}}}}\sqrt{\frac{\gamma _{i}^{\text{UG}}}{\gamma _{i}^{\text{UG}}+1}}\tag{11} \\ 
 & +\sum\limits_{s=1}^{S}{\frac{{{\alpha }_{s,i}}\left[ t \right]A\beta }{d_{s}^{\text{UR}}[t]{{(d_{s,i}^{\text{RG}})}^{\frac{\xi _{s,i}^{\text{RG}}}{2}}}}\sqrt{\frac{\gamma _{s,i}^{\text{RG}}}{\gamma _{s,i}^{\text{RG}}+1}}} \\ 
 & \times [1,{{e}^{j2\pi {{f}_{c}}{{D}_{I}}\sin \theta _{s,i}^{\text{RG}}\cos \varphi _{s,i}^{\text{RG}}/c}}\text{,}...\text{,} \\ 
 & {{e}^{j2\pi {{f}_{c}}({{J}_{I}}-1){{D}_{I}}\sin \theta _{s,i}^{\text{RG}}\cos \varphi _{s,i}^{\text{RG}}/c}}{{]}^{\text{T}}} \\ 
 & \otimes [1,{{e}^{j2\pi {{f}_{c}}{{D}_{b}}\sin \theta _{s,i}^{RG}\sin \varphi _{s,i}^{RG}/c}}\text{,}...\text{,} \\ 
 & {{e}^{j2\pi {{f}_{c}}({{J}_{s,i}}[t]-1){{D}_{b}}\sin \theta _{s,i}^{RG}\sin \varphi _{s,i}^{RG}/c}}{{]}^{\text{T}}}{{\Phi }_{s,i}}[t] \\ 
 & \times [1,{{e}^{-j2\pi {{f}_{c}}{{D}_{I}}\sin \theta _{s}^{UR}[t]\cos \varphi _{s}^{UR}[t]/c}}\text{,}...\text{,} \\ 
 & {{e}^{-j2\pi {{f}_{c}}(J{}_{I}-1){{D}_{I}}\sin \theta _{s}^{UR}[t]\cos \varphi _{s}^{UR}[t]/c}}{{]}^{\text{T}}} \\ 
 & \otimes [1,{{e}^{-j2\pi {{f}_{c}}{{D}_{b}}\sin \theta _{s}^{UR}[t]\sin \varphi _{s}^{UR}[t]/c}}\text{,}...\text{,} \\ 
 & {{e}^{-j2\pi {{f}_{c}}({{J}_{s,i}}[t]-1){{D}_{b}}\sin \theta _{s}^{UR}[t]\sin \varphi _{s}^{UR}[t]/c}}{{]}^{\text{T}}}.
\end{align*}
\indent\textit{From equation (\ref{eq11}), we know that maximizing ${{\left| Q_{i}^{\text{UG}}[t] \right|}^{2}}$ is equivalent to aligning the LoS component of the reflected link with the LoS component of the direct link. Therefore, we obtain the optimal phase control strategy denoted by equation (\ref{eq10}). Theorem 1 is proved.}
\end{proof}
\indent Further, the composite channel power gain between the UAV and user $i$ at time slot $t$ can be denoted by equation (\ref{eq12}) at the bottom of the next page. It can be observed that equation (\ref{eq12}) consists of three components, where the first is the UAV-user direct component, the second is the UAV-IRS-user relevant component, and the third is the composite component.
\begin{figure*}[b]
\hrulefill
	\begin{align*}\label{eq12}
\tiny
\left| Q_{i}^{\text{UG}}\left[ t \right] \right|^2 &=\frac{\beta}{\left( d_{i}^{\text{UG}}\left[ t \right] \right) ^{\xi _{i}^{\text{UG}}}}\frac{\gamma _{i}^{\text{UG}}}{\gamma _{i}^{\text{UG}}+1}+\sum\limits_{\text{s}=1}^{S}{\frac{\alpha _{s,i}\left[ t \right] A^2\beta ^2J_b^2J_{s,i}\left[ t \right] ^2}{d_{s}^{\text{UR}}\left[ t \right] ^2\left( d_{s,i}^{\text{RG}} \right) ^{\xi _{s,i}^{\text{RG}}}}\frac{\gamma _{s,i}^{\text{RG}}}{\gamma _{s,i}^{\text{RG}}+1}}\tag{12}\\
&+\sum\limits_{\text{s}=1}^{S}{\frac{2\alpha _{s,i}\left[ t \right] A\beta ^{3/2}J_bJ_{s,i}\left[ t \right]}{d_{s}^{\text{UR}}\left[ t \right] \left( d_{s,i}^{\text{RG}} \right) ^{\xi _{s,i}^{\text{RG}}/2}\left( d_{i}^{\text{UG}}\left[ t \right] \right) ^{\xi _{i}^{\text{UG}}/2}}\sqrt{\frac{\gamma _{i}^{\text{UG}}}{\gamma _{i}^{\text{UG}}+1}}\sqrt{\frac{\gamma _{s,i}^{\text{RG}}}{\gamma _{s,i}^{\text{RG}}+1}}}.
	\end{align*}
\end{figure*}
\subsection{NOMA Transmission}
\indent We consider utilizing NOMA to transmit information from all users on the same frequency. Without loss of generality, we use ${{w}_{i,j}}[t]\in \{0,1\},\forall i\ne j,i\in N,j\in N$ to denote the relative SIC decoding order of users $i$ and $j$ at time slot $t$, with ${{w}_{i,j}}[t]=1$ denoting that user $i$ decodes before $j$ and vice versa. Thus, the maximum transmission rate of user $i$ is denoted by:
\begin{align*}\label{eq13}
{{R}_{i}}[t]={{\log }_{2}}\left( 1+\frac{{{p}_{i}}[t]|Q_{i}^{\text{UG}}[t]{{|}^{2}}}{|Q_{i}^{\text{UG}}[t]{{|}^{2}}\sum\limits_{j,j\ne i}^{N}{{{w}_{j,i}}[t]}{{p}_{j}}[t]+{{\delta }^{2}}} \right),\tag{13}
\end{align*}
\noindent where ${{\delta }^{2}}$ denotes the power of additive Gaussian white noise, and the transmit power allocated for user $i$ at time slot $t$ is denoted by ${{p}_{i}}[t]$. In addition, in order to guarantee the fairness among transmission rates of users, when user $i$ decodes before user $j$, it is necessary to satisfy:
\begin{align*}\label{eq14}
{{p}_{j}[t]}\ge {{w}_{i,j}}[t]{{p}_{i}[t]},\forall i\ne j\in N.\tag{14}
\end{align*}

\subsection{UAV Energy Consumption}

\indent Considering the flight and communication energy consumption of the UAV, we employ the flight energy consumption model of a rotary-wing UAV~\cite{9043712}, which is denoted by:
\begin{align*}\label{eq15}
{{P}_{Y}}[t]=\mathcal{P}(1+\frac{3\parallel \textbf{v}[t]{{\parallel }^{2}}}{{{\omega }^{2}}{{r}^{2}}})+\frac{\mathcal{W}\mathcal{V}}{\parallel \textbf{v}[t]\parallel }+\frac{1}{2}\mathbb{G}\rho \varsigma \mathcal{M}\parallel \textbf{v}[t]{{\parallel }^{3}},\tag{15}
\end{align*}
\noindent where $\textbf{v}\text{ }\!\![t]\!\!\text{ =}{{\left[ {{v}_{x}}[t],{{v}_{y}}\left[ t \right] \right]}^{\text{T}}}$ is the velocity vector of the UAV. Symbol $\omega$ denotes the angular velocity of blades, and $\mathcal{V}$ denotes the average rotor induction velocity of the UAV in flight. Symbol $r$ denotes the radius of the rotor, and $\varsigma$ denotes the rotor stability in cubic meters, while $\mathcal{M}$ denotes the area of the rotor disk. Environmental factors include symbols $\rho$ and $\mathbb{G}$, which denote air density and the airframe drag ratio, respectively. Symbol $\mathcal{P}$ is the blade profile power in the flying state, and $\mathcal{W}$ denotes induced power in the hovering state. Symbol $\eta$ denotes the correction factor for the realistic power, and then the total communication power consumption can be expressed by ${{P}_{M}[t]}=\eta \sum\limits_{i=1}^{N}{{{p}_{i}[t]}}$. Note that, considering the performance evaluation under different transmit powers in Section IV, as well as the UAV energy consumption model in related works~\cite{9043712,9804220}, this paper introduces the communication energy consumption of UAV. The total power consumption of a UAV can be defined by ${{P}_{\text{sum}}}[t]={{P}_{M}[t]}+{{P}_{Y}}[t]$.

\subsection{Problem Formulation}

\indent Based on the above definitions of transmission rate and UAV power consumption, we define the ratio of the total transmission rate to the total power consumption of the UAV at time slot $t$ as the system energy efficiency, i.e., $\sum\limits_{i=1}^{N}{{R}_{i}[t]}/{P}_{\text{sum}}[t]$. We aim to obtain the maximum system energy efficiency by jointly optimizing user-IRS association, UAV trajectories, power allocation and SIC decoding orders. By defining $\alpha =\left\{ {{\alpha }_{s,i}}[t],\forall i\in N,s\in S,t\in T \right\}$, $w=\left\{{{w}_{i,j}}[t],\forall i\ne j\in N,t\in T \right\}$, $L=\left\{ L\text{ }\!\![\!\!\text{ t }\!\!]\!\!\text{ },\forall t\in T \right\}$ and $p=\left\{ {{p}_{i}}[t],\forall i\in N,t\in T \right\}$, the system energy efficiency maximization problem is formulated as follows:
\begin{align*}\label{eq16}
P0:\text{ }&\underset{\alpha ,w,L,p}{\mathop{\text{maximize}}}\,\text{ }\sum\limits_{t=1}^{T}{\frac{\sum\limits_{i=1}^{N}{{{R}_{i}}[}t]}{{{P}_{\text{sum}}}[t]}},\tag{16}
\\ \text{s.t.}\text{ }& C1:\text{ }{{\alpha }_{s,i}}\left[ t \right]\in \left\{ 0,1 \right\},\forall s\in S,i\in N,t\in T, 
\\& C2:\text{ }\sum\limits_{s=1}^{S}{{{\alpha }_{s,i}}\left[ t \right]}=1,\forall i\in N,t\in T, 
\\& C3:\text{ }\sum\limits_{i=1}^{N}{{{\alpha }_{s,i}}\left[ t \right]}\le {{J}_{b}},\forall s\in S,t\in T, 
\\& C4:\text{ }\left\| L[t+1]-L[t] \right\|\le \tau {{V}_{\max }},\forall t\in T, 
\\& C5:\text{ }L\text{ }\!\![\!\!\text{ 0 }\!\!]\!\!\text{  = }{{L}_{0}}, 
\\& C6:\text{ }{{Z}_{\min }}\le z[t]\le {{Z}_{\max }},\forall t\in T, 
\\& C7:\text{ }{{w}_{i,j}}[t]\in \left\{ 0,1 \right\},\text{ }\forall i\ne j\in N, t\in T, 
\\& C8:\text{ }{{w}_{i,j}}[t]+{{w}_{j,i}}[t]=1,\forall i\ne j\in N, 
\\& C9:\text{ }{{p}_{j}}\ge {{w}_{i,j}}[t]{{p}_{i}},\forall t\in T, i\ne j\in N, 
\\& C10:\text{ }\sum\limits_{i=1}^{N}{p_i}[t]=p_{\max},\forall t\in T, 
\\& C11:\text{ }{{p}_{i}}[t]\ge 0,\forall i\in N, t\in T. 
\end{align*}

\noindent Note that $C1$ defines the range of IRS association variable ${{\alpha }_{s, i}}$, and $C2$ denotes that each user at time slot $t$ should be assigned with an IRS, and $C3$ denotes the number limit of users to be served by one IRS at time slot $t$. $C4$ denotes the limit of the movement distance of the UAV at time slot $t$. $C5-C6$ denote the starting position and flight altitude constraints of the UAV, respectively. $C7$ defines the range of SIC decoding scheduling variable, and $C8$ indicates that when the signal of user $i$ is interfered with by that of user $j$, the latter one is not interfered with by the former one. $C9$ represents that the transmit power allocated to user $i$ is not more than that of user $j$ when user $i$ decodes before user $j$. $C10$ and $C11$ denote the transmit power constraints at time slot $t$. Note that the design of Problem $P0$ assumes perfect channel state information (CSI)~\cite{9681874,9849460,9769985}.
\begin{theorem}
	The formulated optimization Problem $P0$ is NP-hard.
\end{theorem}
\begin{proof}
\textit{The objective function of Problem $P0$ is nonconvex with respect to user association variable ${{\alpha }_{s,i}}[t]$, decoding scheduling variable ${{w}_{i,j}}[t]$ and UAV position $L[t]$, and these variables are highly coupled with each other. Since variables ${{w}_{i,j}}[t]$ and ${{\alpha }_{s,i}}[t]$ are binary, and constraints $C1$ and $C7$ are nonconvex, Problem $P0$ is a mixed-integer fractional nonconvex problem. In fact, even if the position of the UAV is fixed, the problem reduces to finding an optimal solution for user association, decoding order scheduling and power allocation to maximize the system efficiency. The reduced problem belongs to the classical Knapsack problem, a kind of classical NP-hard problem. Thus, Problem $P0$ is also NP-hard. Theorem 2 is proved.}
\end{proof}
\section{A System Energy Efficiency Maximization Scheme}
\begin{figure}
  \centering
  \includegraphics[width=0.70\linewidth]{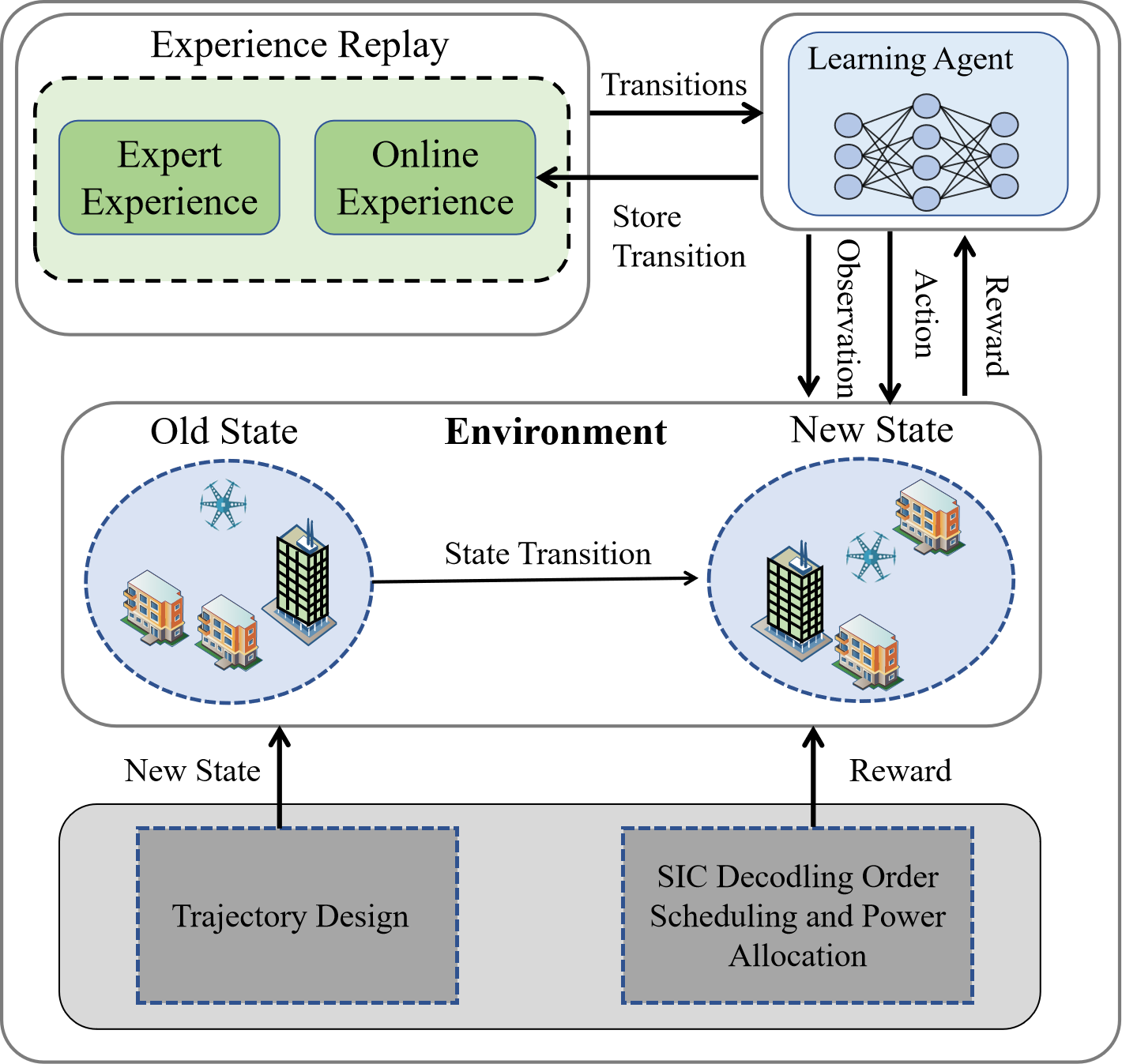}
  \caption{Structure of the proposed algorithm.}
  \label{fig2}
\vspace{-5mm}
\end{figure}
\indent Since Problem $P0$ is NP-hard, finding the optimal solution in polynomial time is impossible. In this section, we develop an alternating algorithm for maximizing the energy efficiency of the system and divide Problem $P0$ into three subproblems based on the coupling of variables. Specifically, given the previous iteration of UAV position, initial NOMA decoding order and power allocation for each iteration, we first develop an inverse soft-Q learning based algorithm to obtain user association for IRSs. Since the SIC decoding order is determined by the current CSI~\cite{9277627,9953122}, and considering the fixed SIC decoding order in related study~\cite{9494520,9839554}, we assume the decoding order among different time slots is independent. To facilitate the design of the subsequent algorithms and to satisfy constraint $C9$, we initialize the SIC decoding order and power allocation at each time slot. Next, for given user association, SIC decoding order and power allocation, optimizing UAV trajectories by integrating SCA and Dinkelbach's method. Finally, we fix the user association of IRSs, and the UAV location, and employ penalty-based SCA to optimize the SIC decoupling order and power allocation.
\subsection{User Association Determination Based on Improved Inverse Soft-Q Learning}

\indent In each iteration, we substitute UAV coordinate $L[t]$, initial SIC decoding order $w^{0}$ and power allocation $p^{0}$ into Problem $P0$. As a result, we can obtain Problem $P1$ by:
\begin{align*}\label{17}
\small
&P1:\underset{\alpha }{\mathop{\text{maximize}}}\,\text{ }\sum\limits_{t=1}^{T}{\frac{\sum\limits_{i=1}^{N}{{{R}_{i}}[}t]}{{{P}_{\text{sum}}}[t]}}, \tag{17}\\ 
&\text{s.t.}\text{ } C1-C3 \text{ }in \text{ }P0. 
\end{align*}
\indent Although merely binary variable $\alpha$ is considered in Problem $P1$, it is still an integer nonlinear programming problem. Problem $P1$ can be traversed by the enumeration method to obtain the optimal solution for energy efficiency. However, the time complexity of the method is exponential with respect to the number of users and IRSs, and the time complexity with $N$ users and $S$ IRSs is ${\mathrm O}({{N}^{S}})$. Therefore, we propose a user association algorithm based on inverse soft-Q learning, which can rapidly converge and approximate optimal solutions in dynamic environments by learning from expert datasets. Compared to deep reinforcement learning (DRL), inverse soft-Q learning approximates the optimal solution more closely over a smaller number of environment interactions.
 
 \indent We can start with a small number of iterations using the traversal method to obtain some results as expert datasets. Then, our proposed method based on inverse soft-Q learning can be used for online learning with expert datasets, which enable the agent to make effective user association decisions. Dataset $\mathcal{D}$ represents expert datasets containing state-action pairs $\left\langle \mathbb{S}\text{ }\!\![\!\!\text{ }t\text{ }\!\!]\!\!\text{ },\mathcal{A}[t] \right\rangle$ to characterize the optimal decisions in these states. Since the expert strategy is unknown, we can learn by sampling the expert dataset to enable the agent to make behavior close to the expert.

\indent Specifically, define action $\mathcal{A}[t]\in\mathfrak{A}$ to represent binary variable $\alpha$ chosen by the agent at state $\mathbb{S}\text{ }\!\![\!\!\text{ }t\text{ }\!\!]\!\!\text{ }\in \mathfrak{S}$, and the state transitions to $\mathbb{S}[t+1]$ after taking the action. Symbol $\mathfrak{A}$ denotes the set of actions, and action $\mathcal{A}[t]=[{{\alpha }_{1,1}}[t],{{\alpha }_{1,2}}[t],...,{{\alpha }_{s,i}}[t],...{{\alpha }_{S,N}}[t]]$. Symbol $\mathfrak{S}$ denotes the state set, and $\mathbb{S}[t]$ denotes the position of the UAV at time slot $t$. According to the optimization objective of maximizing energy efficiency, one-step reward function $\Re[t]$ is set as the energy efficiency resulting from executing the user association policy corresponding to action $\mathcal{A}[t]$, which is denoted by $\Re \text{ }\!\![\!\!\text{ }t\text{ }\!\!]\!\!\text{ =}{\sum\limits_{i=1}^{N}{{{R}_{i}}[}t]}/{{{P}_{\text{sum}}}[t]}$. The pseudo-code for the specific training process is shown by Algorithm 1.
\begin{algorithm}[!h]
    \caption{Low Complexity Training Process of The Learning Agent}
    \label{A1}
    \renewcommand{\algorithmicrequire}{\textbf{Input:}}
    \renewcommand{\algorithmicensure}{\textbf{Output:}}
    \begin{algorithmic}[1]
        \Require 
			Expert dateset $\mathcal{D}$, batch size $B$, learning rates $\mathfrak{L}$, policy $\pi$, initial Q-function $\mathcal{Q}_{\mathfrak{u}}$ and  corresponding parameter $\mathfrak{u}[t]$.
        \Ensure Learned policy $\pi$.
        
        \State Initialize initial state representation $\mathbb{S}[0]$, empty memory and load expert dateset $\mathcal{D}$.

        \For{each $t = 1,2,...,T$}
			\State Obtain $\pi$ through policy network.
			\State Select action $\mathcal{A}[t]$ based on $\pi$ and state $\mathbb{S}[t]$.
			\State Obtain $\mathbb{S}[t+1]$ and $\Re[t]$ by solving Problems $P2$ and $P3$.
			\State Store transition $\{\mathbb{S}[t],\mathcal{A}[t],\mathbb{S}[t+1],\Re[t],t \}$ in the online memory.
            \If {the online memory size equal to the expert memory size}
                \State Sample a batch of state transition $\{\mathbb{S}[t],\mathcal{A}[t],\mathbb{S}[t+1],\Re[t],t \}$ with size $B$.
                \State Train Q-function based on equation (\ref{eq18}) to get $\mathcal{Q}\text{*}_{\mathfrak{u}}$:	
                \State $\mathfrak{u}[t+1] \gets \mathfrak{u}[t] \nabla_{\mathfrak{u}}[-\mathcal{J}\left( \mathcal{Q}_{\mathfrak{u}} \right)]$.
                \State Recover policy $\pi{\text{*}}$:
                \State $\text{ }\!\!\pi\!\!\text{ *=}\frac{1}{\mathcal{Z}}\exp \mathcal{Q}_{\mathfrak{u}}\text{*}$. 
            \EndIf
        \EndFor
    \end{algorithmic}
\end{algorithm}

\indent Fig. 2 illustrates the overall framework diagram for Algorithm 1, which mainly consists of experience replay (ER), the agent, and the environment. Note that ER consists of two parts, the expert experience and the online experience~\cite{garg2021iqlearn}. Before model training, the agent needs to interact with the environment to generate the initial transitions that can be stored in the ER. In each step, the agent chooses an action with the largest Q-value based on the current state, to get a reward from the environment and proceed to the next state. Then, one step transition $\{\mathbb{S}[t]\text{,}\mathcal{A}[t],\mathbb{S}[t+1],\Re [t],t\}$ is generated and stored in ER. The training process starts when there are enough ER transitions. Specifically, a certain batch of transitions is randomly sampled from ER, and the batch of experience consists of both expert experience and online experience.

\indent To begin, we introduce an inverse soft-Q objective function to update the Q-function~\cite{garg2021iqlearn}, which is obtained by transforming the maximum entropy inverse reinforcement learning objective function from the reward-policy space to the Q-policy space. The maximum entropy inverse reinforcement learning objective function aims to find a reward function that confers high rewards to expert policy and low rewards to other policies. As a result, Q-function is updated by:
\begin{align*}\label{eq18}
&\mathcal{J}\left( \mathcal{Q} \right)={{\mathbb{E}}_{{{\text{ }\!\!\pi\!\!\text{ }}_{E}}}}\left[ \mathcal{F}\left( \mathcal{Q}\left( \mathbb{S}\text{,}\mathcal{A} \right)-l{{\mathbb{E}}_{\mathbb{S}'\sim \Pi \left( \left. \cdot \right|\mathbb{S},\mathcal{A} \right)}}\mathbb{V}\left( \mathbb{S}' \right) \right) \right] \tag{18}\\
&-\left( 1-l \right){{\mathbb{E}}_{\text{ }\!\!\pi'\!\!\text{ }}}\left[ \mathbb{V}\left( \tilde{\mathbb{S}} \right) \right], 
\end{align*}
\noindent where symbols ${{\pi }_{E}}$ and $\pi'$ denote the occupancy measures of the expert strategy and the agent's initial strategy, respectively, and $l\in \left( 0,1 \right)$ denotes the discount factor. Function $\mathcal{F}\left( \text{x} \right)=\text{x}-{{\text{x}}^{2}}/4\mathcal{C}$, where $\mathcal{C}>0$ denotes the deflation factor of the original divergence. Variable $\mathcal{Q}\left( \mathbb{S}\text{,}\mathcal{A} \right)$ denotes the Q-value obtained by executing action $\mathcal{A}[t]$ at state $\mathbb{S}[t]$, and $\Pi \left( \left. \cdot \right|\mathbb{S},\mathcal{A} \right)$ denotes the state distribution of the next state by executing action $\mathcal{A}[t]$ at state $\mathbb{S}[t]$. Soft value function is denoted by $\mathbb{V}\left( \mathbb{S}' \right)=\log \sum\nolimits_{\mathcal{A}}{\exp \mathcal{Q}\left( \mathbb{S}\text{ }\!\!'\!\!\text{,}\mathcal{A} \right)}$, where $\mathbb{S}'$ and $\tilde{\mathbb{S}}$ denote the next state and the initial state, respectively.

\indent After updating the Q-function, we recover the policy of the agent from the Q-function. Denote the updated Q-function by $\mathcal{Q}\text{*}$ and the updated policy by $\text{ }\!\!\pi\!\!\text{*}$, which can be expressed by $\text{ }\!\!\pi\!\!\text{ *=}\frac{1}{\mathcal{Z}}\exp \mathcal{Q}\text{*}$, where $\mathcal{Z}$ denotes the normalization factor. Finally, we update the parameters of the agent to enable it to make user association decision $\alpha$ that is closer to the expert.

\subsection{Trajectory Optimization Based on SCA and Dinkelbach}

\indent Problem $P2$ can be obtained by introducing solution $P1$ $\alpha[t]$ of Problem P1, power allocation $p^{0}$ and SIC decoding orders $w^{0}$ into Problem $P0$:
\begin{align*}\label{eq19}
P\text{2: }&\underset{L}{\text{maximize}}\,{\frac{\sum\limits_{i=1}^{N}{{{R}_{i}}[}t]}{ {{P}_{M}[t]}+{{P}_{Y}}[t]}},\tag{19} \\ 
\text{s.t.}\text{ } &C4-C7\text{ }in \text{ }P0.
\end{align*}
\noindent However, Problem $P2$ is still a fractional nonconvex optimization function. Considering that the objective function is fractional, the maximum system energy efficiency can be denoted by a non-negative parameter:
\begin{align*}\label{eq20}
\varpi = {\frac{\sum\limits_{i=1}^{N}{{{R}_{i}}[t]}}{{{P}_{M}[t]}+{{P}_{Y}}[t] }}.\tag{20}
\end{align*}
\indent Then, we can use the Dinkelbach method to transform the objective function (\ref{eq19}) into a simplified form:
\begin{align*}\label{eq21}
\underset{L}{\text{maxmize}}\,{\sum\limits_{i=1}^{N}{R_{i}^{{}}[t]}}-\varpi{\left( {{P}_{M}[t]}+{{P}_{Y}}[t] \right)}.\tag{21}
\end{align*}
\indent When ${{L}^{*}}$ is the optimal solution, the condition for ${{\varpi}^{*}}$ to be the optimal system energy efficiency of Problem $P2$ is given by Theorem 3.
\begin{theorem}
	Symbols ${{\varpi }^{*}}$ and ${{L}^{*}}$ can be considered as the optimal system energy efficiency and the optimal solution of Problem $P2$, respectively, if and only if:
\begin{align*}\label{eq22}
\text{max} \left\{{\sum\limits_{i=1}^{N}{R_{i}^{{}}[t]}}-{{\varpi }^{*}}{\left( {{P}_{M}[t]}+{{P}_{Y}}[t] \right)}\right\}=0 .\tag{22}
\end{align*}
\end{theorem}
\begin{proof}
\textit{Define $\mathfrak{F}\left( \varpi  \right)={\sum\limits_{i=1}^{N}{R_{i}^{{}}[t]}}-\varpi{\left( {{P}_{M}[t]}+{{P}_{Y}}[t] \right)}$, $\mathfrak{j}\left( L \right)={\sum\limits_{i=1}^{N}{{{R}_{i}}[t]}}$ and $\mathfrak{k}\left( L \right)= {{P}_{M}[t]}+{{P}_{Y}}[t] $. Thus, we have $\frac{\mathfrak{j}\left( {{L}^{*}} \right)}{\mathfrak{k}\left( {{L}^{*}} \right)}={{\varpi }^{*}}\Rightarrow \mathfrak{j}\left( {{L}^{*}} \right)+{{\varpi }^{*}}\mathfrak{k}\left( {{L}^{*}} \right)=0$. Let $\tilde{L}$ be the arbitrary feasible solution of Problem $P2$, and we have $\frac{\mathfrak{j}\left( {\tilde{L}} \right)}{\mathfrak{k}\left( {\tilde{L}} \right)}\le {{\varpi }^{*}}\Rightarrow \mathfrak{j}\left( {\tilde{L}} \right)+{{\varpi }^{*}}\mathfrak{k}\left( {\tilde{L}} \right)\le 0$. Theorem 3 is proved.}
\end{proof}
\indent Therefore, according to Theorem 3, the optimal solutions of Problems $P2$ and $P2'$ are equivalent.
\begin{align*}\label{eq23}
P2':\text{ } &\underset{L}{\text{maxmize}}\,{\sum\limits_{i=1}^{N}{R_{i}^{{}}[t]}}-\varpi {\left({{P}_{M}[t]}+{{P}_{Y}}[t] \right)},\tag{23}\\ 
& \text{s.t.}\text{ } C4-C7\text{ }in \text{ }P0.
\end{align*}
\indent However, Problem $P2'$ is still nonconvex. First, we deal with the nonconvexity of ${{\left| Q_{i}^{\text{UG}}[t] \right|}^{2}}$ in $R_{i}[t]$. For tractable solutions, we reformulate it by:
\begin{align*}\label{eq24}
{{\left| Q_{i}^{\text{UG}}[t] \right|}^{2}}&=\frac{\mathbb{A}}{{{(d_{i}^{\text{UG}}[t])}^{\xi _{i}^{\text{UG}}}}}+\sum\limits_{\text{s}=1}^{S}{\frac{\mathbb{B}}{d_{s}^{\text{UR}}{{[t]}^{2}}}} +\tag{24}\\
&\sum\limits_{\text{s}=1}^{S}{\frac{\mathbb{D}}{d_{s}^{\text{UR}}[t]{{(d_{i}^{\text{UG}}[t])}^{\xi _{i}^{\text{UG}}/2}}}},
\end{align*}
\noindent where $\mathbb{A}=\beta \gamma _{i}^{\text{UG}}/(\gamma _{i}^{\text{UG}}+1)$, $\mathbb{B}=\frac{{{\alpha }_{s,i}}\left[ t \right]\gamma _{s,i}^{\text{RG}}{{A}^{2}}{{\beta }^{2}}{{J}_{b}}^{2}{{J}_{s,i}}{{\left[ t \right]}^{2}}}{\gamma _{\text{s,i}}^{\text{RG}}+1{{(d_{s,i}^{\text{RG}})}^{\xi _{s,i}^{\text{RG}}}}}$ and $\mathbb{D}=\frac{2{{\alpha }_{s,i}}\left[ t \right]A{{\beta }^{3/2}}{{J}_{b}}{{J}_{s,i}}\left[ t \right]}{{{(d_{s,i}^{\text{RG}})}^{\xi _{s,i}^{\text{RG}}/2}}}\sqrt{\frac{\gamma _{i}^{\text{UG}}}{\gamma _{i}^{\text{UG}}+1}}\sqrt{\frac{\gamma _{s,i}^{\text{RG}}}{\gamma _{s,i}^{\text{RG}}+1}}$. Further, we introduce slack variable $\{\mathcal{G}{{}_{i}}[t]=d_{i}^{\text{UG}}[t],\forall i\in N,t\in T\}$ and $\{{{\mathcal{R}}_{s}}[t]={d_{s}^{\text{UR}}},\forall s\in S,t\in T\}$ to further simplify Problem $P2$ for subsequent convexification, and equation (\ref{eq24}) can be converted into:
\begin{align*}\label{eq25}
{{\left| \tilde Q_{i}^{\text{UG}}[t] \right|}^{2}}&=\frac{\mathbb{A}}{\mathcal{G}{{}_{i}}[t]{{}^{\xi _{i}^{\text{UG}}}}}+\sum\limits_{s=1}^{S}{\frac{\mathbb{B}}{{{\mathcal{R}}_{s}}[t]{{}^{2}}}}+ \tag{25}\\
&\sum\limits_{s=1}^{S}{\frac{\mathbb{D}}{\mathcal{G}{{}_{i}}[t]{^{\xi _{i}^{\text{UG}}/2}}{{\mathcal{R}}_{s}}[t]}}.
\end{align*}
\indent To facilitate the following calculations and convexification, introducing auxiliary variables $\{{{\mathcal{H}}_{i}}[t]= \sum\limits_{j,j\ne i}^{N}{{{w}_{i,j}}[t]}{{p}_{j}}[t]+\frac{{{\delta }^{2}}}{{\left| \tilde Q_{i}^{\text{UG}}[t] \right|}^{2}},\forall i\in N,t\in T\}$, we can obtain:
\begin{align*}\label{eq26}
P2'':\text{ }&\underset{L,\mathcal{G},\mathcal{R},\mathcal{H},\mathcal{B}}{\mathop{\text{maxmize}}}\,{\sum\limits_{i=1}^{N}{{{\log }_{2}}\left( 1+\frac{{{p}_{i}}[t]}{{{\mathcal{H}}_{i}}[t]} \right)}}-\varpi {{P}_{\text{sum}}}[t], \tag{26}\\ 
&\text{s.t.}\text{ }C4-C6\text{ }in \text{ }P0, \\ 
 & \text{ }\text{ }\text{ }\text{ } \text{ }C12:\text{ } \mathcal{G}{{}_{i}}[t]{{}^{2}}\ge {{\left\| L[t]-{{L}_{i}} \right\|}^{2}},\forall i \in N,t \in T, \\ 
 & \text{ }\text{ }\text{ }\text{ } \text{ }C13:\text{ }{{\mathcal{R}}_{s}}[t]{{}^{2}}\ge {{\left\| L[t]-{{L}_{s}} \right\|}^{2}},\forall s \in S,t \in T. 
\end{align*}

\begin{algorithm}[t]
    \caption{Low Complexity Trajectory Optimization Algorithm}
    \label{A2}
    \renewcommand{\algorithmicrequire}{\textbf{Input:}}
    \renewcommand{\algorithmicensure}{\textbf{Output:}}
    \begin{algorithmic}[1]
        \Require UAV coordinate $L[t]$, and user association decisions $\{\alpha_{s,i}[t]\}_{s\in S,i \in N}$ from Problem $P1$, the maximum iteration numbers of main loop $\mathbb{Y}$ and inner loop $\mathbb{X}$.
        \Ensure UAV coordinate $L^{*}[t+1]$.
        \State Initialize the coordinate of UAV $L^{\varDelta }[t]=L^{\tilde{\varDelta}}[t]=L[t]$, $\varpi ^{ \varDelta }=\varpi ^{ \tilde{\varDelta}} = 0$, and iteration index of main loop $\varDelta = 1$ and inner loop $\tilde{\varDelta} = 1$, with the convergence tolerance $\epsilon_1 \rightarrow 0$ and $\epsilon_2 \rightarrow 0$.
		\State Main Loop: Dinkelbach
        \For{$\varDelta=1,2,...,\mathbb{Y}$}
		\State Inner Loop: SCA
		\For{$\tilde{\varDelta}=1,2,...,\mathbb{X}$}
\State  Obtain $\{\mathcal{G}{{}_{i}}^{{\tilde{\Delta }}}[t],{{\mathcal{R}}_{s}}^{{\tilde{\Delta }}}[t],\mathcal{H}_{i}^{{\tilde{\Delta }}}[t],\mathbb{V}^{{\tilde{\Delta }}}[t]\}$ by the definitions of auxiliary variables based on $\{L^{{\tilde{\Delta }}}[t], \{p_i^{0}[t]\}_{i \in N}, \{w_{i,j}^{0}[t]\}_{i,j \in N}, \{\alpha_{s,i}[t]\}_{s\in S,i \in N}\}$.

\State Given $\{L^{{\tilde{\Delta }}}[t],\mathcal{G}{{}_{i}}^{{\tilde{\Delta }}}[t],{{\mathcal{R}}_{s}}^{{\tilde{\Delta }}}[t],\mathcal{H}_{i}^{{\tilde{\Delta }}}[t],\mathbb{V}^{{\tilde{\Delta }}}[t]\}$, obtain $L^{\tilde{\varDelta}+1}[t]$ by sloving \textbf{$P2''''$}.
\State Obtain $\varpi^{\tilde{\varDelta}+1}$ by introducing $L^{\tilde{\varDelta}+1}[t]$ in equation (20).

 \If {$ R\left( L^{\tilde{\varDelta}+1}[t] \right) -\varpi ^{\tilde{\varDelta}} P_{sum} \left( L^{\tilde{\varDelta}+1}[t] \right) - 
R\left( L^{\tilde{\varDelta}}[t] \right) +\varpi ^{\tilde{\varDelta}} P_{sum} \left( L^{\tilde{\varDelta}}[t] \right)
\le \epsilon _2$}	
\State $L^{\varDelta+1}[t] = L^{\tilde{\varDelta}+1}[t]$, $\varpi ^{ \varDelta +1}=\varpi ^{ \tilde{\varDelta}+1}$
\State break
\EndIf  
       \EndFor 
		\If {$  R\left( L^{\varDelta +1}\left[ t \right] \right) -\varpi ^{ \varDelta}\left( P_M \left( L^{\varDelta +1}\left[ t \right] \right) +P_Y \left( L^{ \varDelta +1}\left[ t \right] \right) \right) \le \epsilon _1$	}	
\State $ L^{*}[t+1] = L^{\varDelta+1 }[t]$
\State break
\EndIf  
        \EndFor     
    \end{algorithmic}
\end{algorithm}

\indent Next, we continue to convexify Problem $P2''$ by employing SCA. In the internal iteration of SCA, given feasible solution $\{L^{{\tilde{\Delta }}}[t],\mathcal{G}{{}_{i}}^{{\tilde{\Delta }}}[t],{{\mathcal{R}}_{s}}^{{\tilde{\Delta }}}[t]\}$ for the $\tilde{\Delta }$ iteration, we can get $\{L^{{\tilde{\Delta }}}[t],\mathcal{G}{{}_{i}}^{{\tilde{\Delta }}}[t],{{\mathcal{R}}_{s}}^{{\tilde{\Delta }}}[t],\mathcal{H}_{i}^{{\tilde{\Delta }}}[t]\}$. Then, we can obtain the lower bound at the point by using the first-order Taylor expansion as follows:
\begin{align*}\label{eq27}
{{R}_{i}}[t]\ge {{\tilde{R}}_{i}}[t]&={{\log }_{2}}\left( 1+\frac{{{p}_{i}}[t]}{\mathcal{H}_{i}^{{\tilde{\Delta }}}[t]} \right)\tag{27}\\
&-\frac{{{p}_{i}}[t]{{\log }_{2}}(e)}{\mathcal{H}_{i}^{{\tilde{\Delta }}}[t](\mathcal{H}_{i}^{{\tilde{\Delta }}}[t]+{{p}_{i}}[t])}({{\mathcal{H}}_{i}}[t]-\mathcal{H}_{i}^{{\tilde{\Delta }}}[t]), 
\end{align*}
\noindent and the relaxation of constraints:
\begin{align*}
&C14:\text{ }\mathcal{G}{{}_{i}}^{{\tilde{\Delta }}}[t]{{}^{2}}+2\mathcal{G}{{}_{i}}^{{\tilde{\Delta }}}[t]\mathcal{G}{{}_{i}}[t]-\mathcal{G}{{}_{i}}^{{\tilde{\Delta }}}[t]\ge d_{i}^{\text{UG}}{{[t]}^{2}},\\
&\text{ }\text{ }\text{ }\text{ }\text{ }\text{ }\forall i\in N,t\in T, \\ 
&C15:\text{ }{{\mathcal{R}}^{{\tilde{\Delta }}}}_{s}[t]{{}^{2}}+2{{\mathcal{R}}^{{\tilde{\Delta }}}}_{s}[t]{{\mathcal{R}}_{s}}[t]-{{\mathcal{R}}^{{\tilde{\Delta }}}}_{s}[t]\ge d_{s}^{\text{UR}}{{[t]}^{2}},\\
&\text{ }\text{ }\text{ }\text{ }\text{ }\text{ }\forall s\in S,t\in T, \\ 
&C16:\text{ }{{\left| \tilde Q_{i}^{\text{UG}}[t] \right|}^{2}}\le {{\mathcal{B}}_{i}}[t],\forall i\in N,t\in T, 
\end{align*}
\noindent where ${{\mathcal{B}}_{i}}'[t]$ can be denoted by equation (\ref{eq28}) at the bottom of the next page.

\begin{figure*}[b]
\hrulefill
	\begin{align*}\label{eq28}
{{\mathcal{B}}_{i}}[t]&=\frac{\mathbb{A}}{\mathcal{G}_{i}^{{\tilde{\Delta }}}[t]{{}^{\xi _{i}^{UG}}}}+\sum\limits_{s=1}^{S}{\frac{\mathbb{B}}{\mathcal{R}_{s}^{{\tilde{\Delta }}}[t]{{}^{2}}}}+\sum\limits_{s=1}^{S}{\frac{\mathbb{D}}{{{(\mathcal{G}_{i}^{{\tilde{\Delta }}}[t])}^{\xi _{i}^{UG}/2}}\mathcal{R}_{s}^{{\tilde{\Delta }}}[t]}}-\frac{\xi _{i}^{UG}\mathbb{A}}{{{(\mathcal{G}_{\text{i}}^{{\tilde{\Delta }}}[t])}^{\xi _{i}^{UG}+1}}}{({\mathcal{G}}_{i}}[t]-\mathcal{G}_{i}^{{\tilde{\Delta }}}[t]) -\sum\limits_{s=1}^{S}{\frac{2\mathbb{B}}{{{(\mathcal{R}_{s}^{{\tilde{\Delta }}}[t])}^{3}}}{({\mathcal{R}}_{s}}[t]-\mathcal{R}_{s}^{{\tilde{\Delta }}}[t])}\\
&-\sum\limits_{s=1}^{S}{\frac{\xi _{i}^{UG}\mathbb{D}/2}{{{(\mathcal{G}_{i}^{{\tilde{\Delta }}}[t])}^{\xi _{i}^{UG}/2+1}}\mathcal{R}_{s}^{{\tilde{\Delta }}}[t]}(\mathcal{G}_{i}^{{}}[t]-\mathcal{G}_{i}^{{\tilde{\Delta }}}[t])} -\sum\limits_{s=1}^{S}{\frac{\mathbb{D}}{{{(\mathcal{R}_{s}^{{\tilde{\Delta }}}[t])}^{2}}\mathcal{G}_{i}^{{\tilde{\Delta }}}[t]}(\mathcal{R}{{}_{s}}[t]-\mathcal{R}_{s}^{{\tilde{\Delta }}}[t])}.\tag{28} 
	\end{align*}
\end{figure*}

\indent Now, the nonconvexity induced by $R_{i}[t]$ has been solved, and Problem $P2''$ can be approximated by:
\begin{align*}\label{eq29}
P2''':\text{ }&\underset{L,\mathcal{G},\mathcal{R},\mathcal{H},\mathcal{B}}{\mathop{\text{maxmize}}}\,{\sum\limits_{i=1}^{N}{{{{\tilde{R}}}_{i}}[t]}}-\varpi{\left( {{P}_{M}}+{{P}_{Y}}[t] \right)},\tag{29} \\ 
&\text{s.t.}\text{ }C4-C6\text{ }in \text{ }P0,C14-C16.  
\end{align*}
\indent Then, we tackle the nonconvexity of $P_{Y}[t]$ by introducing slack variable $\{\mathbb{V}[t]=\left\|(L[t]- L^{{\tilde{\Delta }}}[t])/\tau\right\|,\forall t\in T\}$ and constraint $C17$:
\begin{align*}\label{eq30}
C17:\text{ }{{\left\| \textbf{v}[t] \right\|}^{2}}\ge {{\mathbb{V}}^{2}}[t],\forall t\in T, \tag{30}
\end{align*}
\noindent and we can get $\tilde P_{Y}[t]$ by:
\begin{align*}\label{eq31}
\small
\tilde {P}_{Y}[t]=\mathcal{P}(1+\frac{3\parallel \textbf{v}[t]{{\parallel }^{2}}}{{{\omega }^{2}}{{r}^{2}}})+\frac{\mathcal{W}\mathcal{V}}{ \mathbb{V}[t]}+\frac{1}{2}\mathbb{G}\rho \varsigma \mathcal{M}\parallel \textbf{v}[t]{{\parallel }^{3}}.\tag{31}
\end{align*}
\indent By using the first-order Taylor expansion at feasible solution $\text{ }\!\!\{\!\!\text{ }{{\textbf{v}}^{{\tilde{\Delta }}}}[t]\}$ for the $\tilde{\Delta }$ iteration, $C17$ is approximated as $C18$ and Problem $P2'''$ is approximated as follows:
\begin{align*}\label{eq32}
\small
P2'''':&\underset{L,\mathcal{G},\mathcal{R},\mathcal{H},\mathcal{B},\mathbb{V}}{\mathop{\text{maxmize}}}\,{\sum\limits_{i=1}^{N}{{{{\tilde{R}}}_{i}}[t]}}-\varpi \left({{P}_{M}}[t]+{{P}_{Y}}[t]\right), \tag{32}\\ 
 \text{s.t.}\text{ }& C4-C6\text{ }in \text{ }P0,C14-C16,\\
 &C18:\text{ }{{\left\| {{\textbf{v}}^{{\tilde{\Delta }}}}[t] \right\|}^{2}}+2{{[{{\textbf{v}}^{{\tilde{\Delta }}}}[t]]}^{T}}(\textbf{v}[t]-{{\textbf{v}}^{{\tilde{\Delta }}}}[t])\ge {{\mathbb{V}}^{2}}[t],\\
&\text{ }\text{ }\text{ }\text{ }\text{ }\text{ }\text{ }\forall t\in T . 
\end{align*}
\indent Problem $P2''''$ is a convex optimization problem and the optimal solution can be obtained by CVX~\cite{9454446}. Note that the optimal solution of Problem $P2''''$ can provide a lower bound for Problem $P2$. The optimization process of Problem $P2$ is detailed in Algorithm 2.
\subsection{Joint SIC Decoupling Order and Power Allocation Optimization Based on Penalty and SCA}

\indent Finally, we substitute results of Problems $P1$ and $P2$ into $P0$, while $P3$ can be expressed as:
\begin{align*}\label{eq33}
P3:\text{ }&\underset{p,w}{\text{maximize}}\, \frac{\sum\limits_{i=1}^{N}{{{R}_{i}}[}t]}{{{P}_{M}[t]}+{{P}_{Y}}[t]}, \tag{33}\\ 
 & \text{s.t.}\text{ }C7-C11\text{ }in \text{ }P0. 
\end{align*}
\indent Since the denominator of Problem $P3$ is constant in the subsequent solution, it is equivalent to solving the following problem:
\begin{align*}\label{eq34}
P3':\text{ }&\underset{p,w}{\text{maximize}}\,{\sum\limits_{i=1}^{N}{{{R}_{i}}[}t]}, \tag{34}\\ 
 &\text{s.t.}\text{ }C7-C11\text{ }in \text{ }P0. 
\end{align*}
\indent First, we tackle the nonconvexity of the objective function caused by binary variable ${{w}_{i,j}}[t]$ by relaxing it to $[0,1]$, and we can obtain constraint $C19$:
\begin{align*}\label{eq35}
C19:\text{ }0\le {{w}_{i,j}}[t]\le 1,\text{ }\forall i\ne j\in N,t\in T. \tag{35}
\end{align*}
\indent To ensure that continuous variable ${{w}_{i,j}}[t]$ still holds value ``0" or ``1", we introduce a penalty term to Problem $P3'$ and denote the channel gain of user $i$ by $\mathbb{K}=|Q_{s,i}^{\text{UG}}[t]{{|}^{2}}$ to facilitate subsequent calculations. Thus, Problem $P3'$ can be transformed to Problem $P3''$ at the bottom of the next page.

\begin{figure*}[b]
\hrulefill
	\begin{align*}\label{eq36}
P3'':\text{ }&\underset{\textbf{p,w}}{\text{maximize}}\,\text{ }{\sum\limits_{i=1}^{N}{\left( {{\log }_{2}}\left( \mathbb{K}\sum\limits_{j,j\ne i}^{N}{{{w}_{j,i}}[t]}{{p}_{j}}[t]+{{\delta }^{2}}+{{p}_{i}}[t]\mathbb{K} \right)-{{\log }_{2}}\left( \mathbb{K}\sum\limits_{j,j\ne i}^{N}{{{w}_{j,i}}[t]}{{p}_{j}}[t]+{{\delta }^{2}} \right) \right)}} \tag{36}\\ 
 & \text{}-\zeta {\sum\limits_{i=1}^{N}{\sum\limits_{j,j\ne i}^{N}{({{w}_{i,j}}[t]}}}-{{({{w}_{i,j}}[t])}^{2}}), \\ 
 &\text{s.t.}\text{ }C8-C11\text{ }in \text{ }P0, C19.
	\end{align*}
\end{figure*}

\indent In the following, we convexify Problem $P3''$ by employing SCA to address the nonconvexity induced by the objective function and constraint $C9$. Let $\left\{ p_{i}^{{\tilde{\tilde{\Delta}}}}[t], {{w}^{{\tilde{\tilde{\Delta}}}}}_{i,j}[t] \right\}$ denote the feasible solution in iteration ${\tilde{\tilde{\Delta}}}$ of Problem $P3$. The lower bound of Problem $P3''$ at that point, i.e., Problem $P3'''$, is obtained by the first-order Taylor expansion, which is shown at the bottom of the next page. Constraint $C9$ is approximated by:
\begin{align*}\label{eq37}
C20:\text{ }{{p}_{j}[t]}\ge &{{w}^{{\tilde{\tilde{\Delta }}}}}_{i,j}[t]{{p}^{{\tilde{\tilde{\Delta }}}}}_{i}[t]+{{w}^{{\tilde{\tilde{\Delta }}}}}_{i,j}[t]\left ({{p}_{i}}[t]-{{p} ^{{\tilde{\tilde{\Delta }}}}}_{i}[t] \right )+\\
&{{p}^{{\tilde{\tilde{\Delta }}}}}_{i}[t] \left ({{w}_{i,j}}[t]-{{w}^{{\tilde{\tilde{\Delta }}}}}_{i,j}[t]\right ),\forall i\ne j\in N, t\in T.\tag{37}
\end{align*}
\indent Then, Problem $P3'''$ is convex and can be solved by the CVX.
\begin{figure*}[b]
\hrulefill
	\begin{align*}\label{eq38}
\tiny
P3''':\text{ } &\underset{\textbf{p,w}}{\text{maximize}}\,{\sum\limits_{i=1}^{N}{\left( {{\log }_{2}}\left( \sum\limits_{j,j\ne i}^{N}{{{w}^{{\tilde{\tilde{\Delta }}}}}_{i,j}[t]}{{p}^{{\tilde{\tilde{\Delta }}}}}_{j}[t]+\frac{{\delta }^{2}}{\mathbb{K}}+{{p}^{{\tilde{\tilde{\Delta }}}}}_{i}[t] \right)-{{\log }_{2}}\left( \sum\limits_{j,j\ne i}^{N}{{{w}^{{\tilde{\tilde{\Delta }}}}}_{i,j}[t]}p_{j}^{{\tilde{\tilde{\Delta }}}}[t]+\frac{{\delta }^{2}}{\mathbb{K}} \right) \right)}} \tag{38}\\ 
 & \text{ }+{\sum\limits_{i=1}^{N}{\frac{\sum\limits_{j,j\ne i}^{N}{{{w}^{{\tilde{\tilde{\Delta }}}}}_{j,i}[t]}\left( p_{j}^{{}}[t]-p_{j}^{{\tilde{\tilde{\Delta }}}}[t] \right)+{{p}^{{}}}_{i}[t]-{{p}^{{\tilde{\tilde{\Delta }}}}}_{i}[t]+\sum\limits_{j,j\ne i}^{N}{p_{j}^{{\tilde{\tilde{\Delta }}}}[t]}\left( {{w}_{j,i}}[t]-{{w}^{{\tilde{\tilde{\Delta }}}}}_{j,i}[t] \right)}{\sum\limits_{j,j\ne i}^{N}{{{w}^{{\tilde{\tilde{\Delta }}}}}_{j,i}[t]}p_{j}^{{\tilde{\tilde{\Delta }}}}[t]+\frac{{\delta }^{2}}{\mathbb{K}}+{{p}^{{\tilde{\tilde{\Delta }}}}}_{i}[t]}{{\log }_{2}}e}} \\ 
 & \text{  }-{\sum\limits_{i=1}^{N}{\frac{\sum\limits_{j,j\ne i}^{N}{{{w}^{{\tilde{\tilde{\Delta }}}}}_{j,i}[t]}\left( p_{j}^{{}}[t]-p_{j}^{{\tilde{\tilde{\Delta }}}}[t] \right)+\sum\limits_{j,j\ne i}^{N}{p_{j}^{{\tilde{\tilde{\Delta }}}}[t]}\left( {{w}_{j,i}}[t]-{{w}^{{\tilde{\tilde{\Delta }}}}}_{j,i}[t] \right)}{\sum\limits_{j,j\ne i}^{N}{{{w}^{{\tilde{\tilde{\Delta }}}}}_{j,i}[t]}p_{j}^{{\tilde{\tilde{\Delta }}}}[t]+\frac{{\delta }^{2}}{\mathbb{K}}}{{\log }_{2}}e}} \\ 
 & \text{  }-\zeta \sum\limits_{i=1}^{N}{\sum\limits_{j,j\ne i}^{N}{({{w}^{{\tilde{\tilde{\Delta }}}}}_{i,j}[t]-{{({{w}^{{\tilde{\tilde{\Delta }}}}}_{i,j}[t])}^{2}}+2{{w}^{{\tilde{\tilde{\Delta }}}}}_{i,j}[t]({{w}_{i,j}}[t]-{{w}^{{\tilde{\tilde{\Delta }}}}}_{i,j}[t]))}},\\
 &\text{s.t.}\text{ }C8, C10, C11\text{ }in \text{ }P0, C19, C20.
	\end{align*}
\end{figure*}
\subsection{Computational Complexity and Convergence Analysis of AISLE Algorithm}

\indent We solve Problems $P1$, $P2$ and $P3$ alternately until convergence, and the convergence of the proposed algorithm is proved by Theorem 4.
\begin{theorem}
The AISLE algorithm is convergent.
\end{theorem}
\begin{proof}
\textit{The convergence of the overall algorithm can be guaranteed by the iteration of the AISLE algorithm. Locally, the inner loop to solve Problems $P2$ and $P3$ are based on SCA-based algorithms. The SCA guarantees that the feasible solution converges if the initial solution is feasible~\cite{10044705}, i.e., assuming that there is a feasible solution $\{\mathcal{X}\}$, and a new feasible solution $\{\mathcal{X'}\}$ can be obtained by executing SCA with iterations until the difference between the objective function values of them is less than the convergence tolerance, i.e., $\mathbb {R}\{\mathcal{X}\}- \mathbb {R}\{\mathcal{X'}\} \le \epsilon$, where $\{\mathbb{R}\}$ denotes the objective function and $\epsilon$ denotes convergence tolerance. In fact, there is an upper bound on the transmission rate for a fixed transmit power, so that the convergence is guaranteed. In addition, from Theorem 3, the Dinkelbach method in the outer layer to solve Problem $P2$ guarantees that the objective value monotonically increases up to a smooth point. Therefore, the whole algorithm is guaranteed to converge. Theorem 4 is proved.}
\end{proof}
\indent The time complexity of the AISLE algorithm is given by Theorem 5.
\begin{theorem}
The total time complexity of the AISLE algorithm can be expressed by:
\begin{align*}\label{eq39}
\small
&{\mathrm O}\Biggl(T(\sum\limits_{k=1}^{\mathcal{K}}{M_{k}\times M_{k+1}}+\mathbb{M}\mathbb{N}^{2}\mathbb{X}\mathbb{Y}( \sqrt{\mathbb{N}}\log (\frac{1}{\epsilon }))+\\
&\tag{39}\mathbb{Q}\mathbb{P}^{2}\mathbb{Z}( \sqrt{\mathbb{P}}\log (\frac{1}{\epsilon '})))\Biggr).
\end{align*}
\end{theorem}
\begin{proof}
\textit{The main complexity of Algorithm 1 is training complexity, i.e., the configuration of the neural network. Let $\mathcal {K}$ and $M_{k}$ denote the number of neuron layers and the neuron with index $k$, respectively. Then, the total complexity of Algorithm 1 is $\sum\limits_{k=1}^{\mathcal{K}}{M_{k}\times M_{k+1}}$. For Algorithm 2, the number of variables and constraints of Problem $P2$ are denoted by $\mathbb{N}=2+3N+S$ and $\mathbb{M}=3+2N+S$, respectively. Let $\mathbb{X}$ denote the iteration number of the inner loop and ${{\epsilon }_{1}}$ express the solution accuracy. Then the algorithmic complexity of the inner loop of Algorithm 2 is ${\mathrm O}\left( \mathbb{M}{{\mathbb{N}}^{2}}\mathbb{X}\left( \sqrt{\mathbb{N}}\log (\frac{1}{\epsilon }) \right) \right)$. By using $\mathbb{Y}$ to denote the iteration number of main loop, the total complexity of Algorithm 2 is ${\mathrm O}\left( \mathbb{M}{{\mathbb{N}}^{2}}\mathbb{X}\mathbb{Y}\left( \sqrt{\mathbb{N}}\log (\frac{1}{\epsilon }) \right) \right)$. Similarly, the complexity of the process for Problem $P3'''$ is ${\mathrm O}\left( \mathbb{Q}{{\mathbb{P}}^{2}}\mathbb{Z}\left( \sqrt{\mathbb{P}}\log (\frac{1}{\epsilon '}) \right) \right)$, where $\mathbb{P}={{N}^{2}}+N$ denotes the number of variables of Problem $P3$, and $\mathbb{Q}=3({{N}^{2}}-N)+N+1$ denotes the number of constraints of Problem $P3$, while $\mathbb{Z}$ and $\epsilon'$ denote the number of iterations as well as the solution accuracy of Problem $P3$, respectively. Therefore, the total time complexity of the proposed AISLE algorithm can be expressed by equation (\ref{eq39}), which is proportional to $N$ and $S$. Equation (\ref{eq39}) is polynomial and has low time complexity. Theorem 5 is proved.}
\end{proof}

\section{Numerical Results}
\begin{figure}
  \centering
  \includegraphics[width=0.60\linewidth]{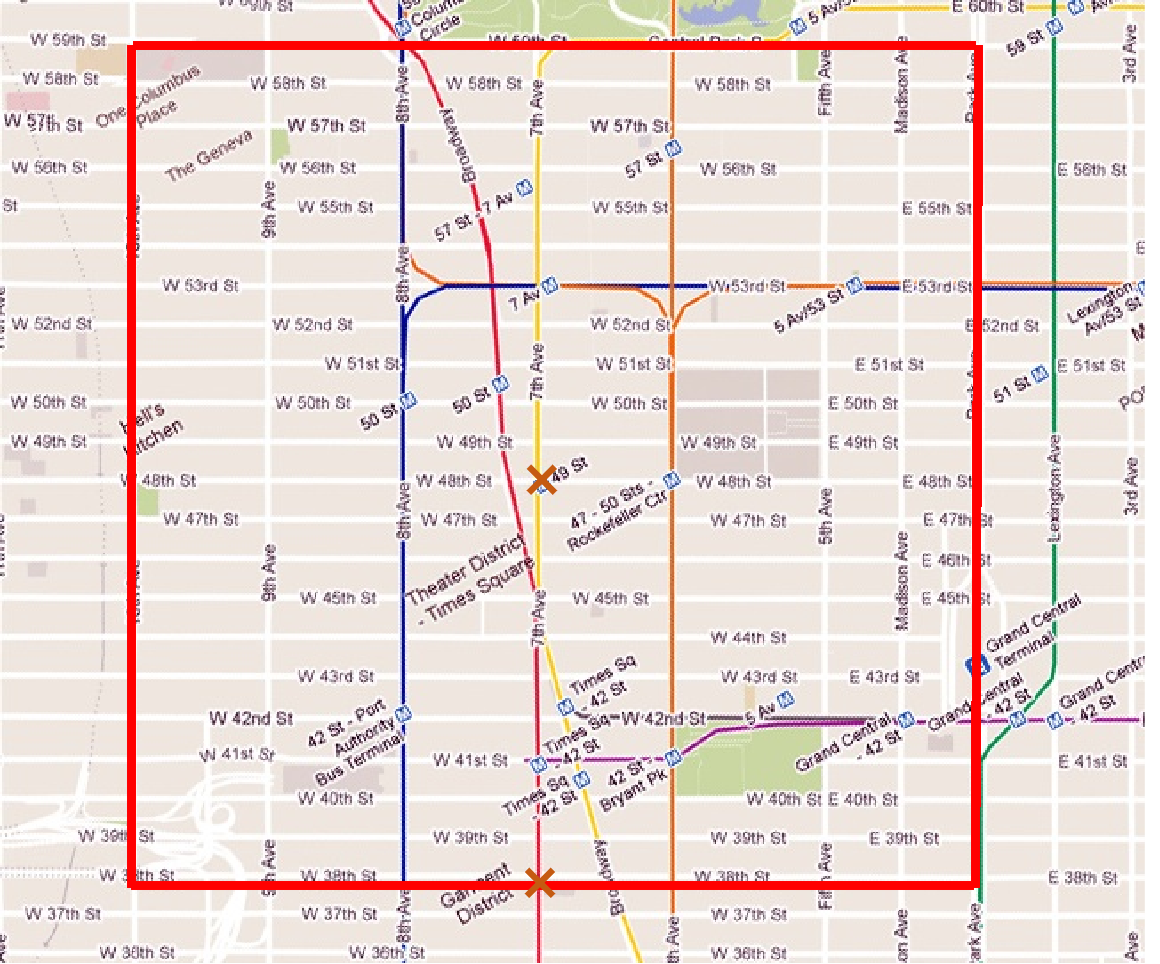}
  \caption{Manhattan city map.}

  \label{fig3}
\vspace{-5mm}
\end{figure}

\begin{figure*}[h]
    \centering

    \subfigure[2000]{\includegraphics[width=0.25\textwidth]{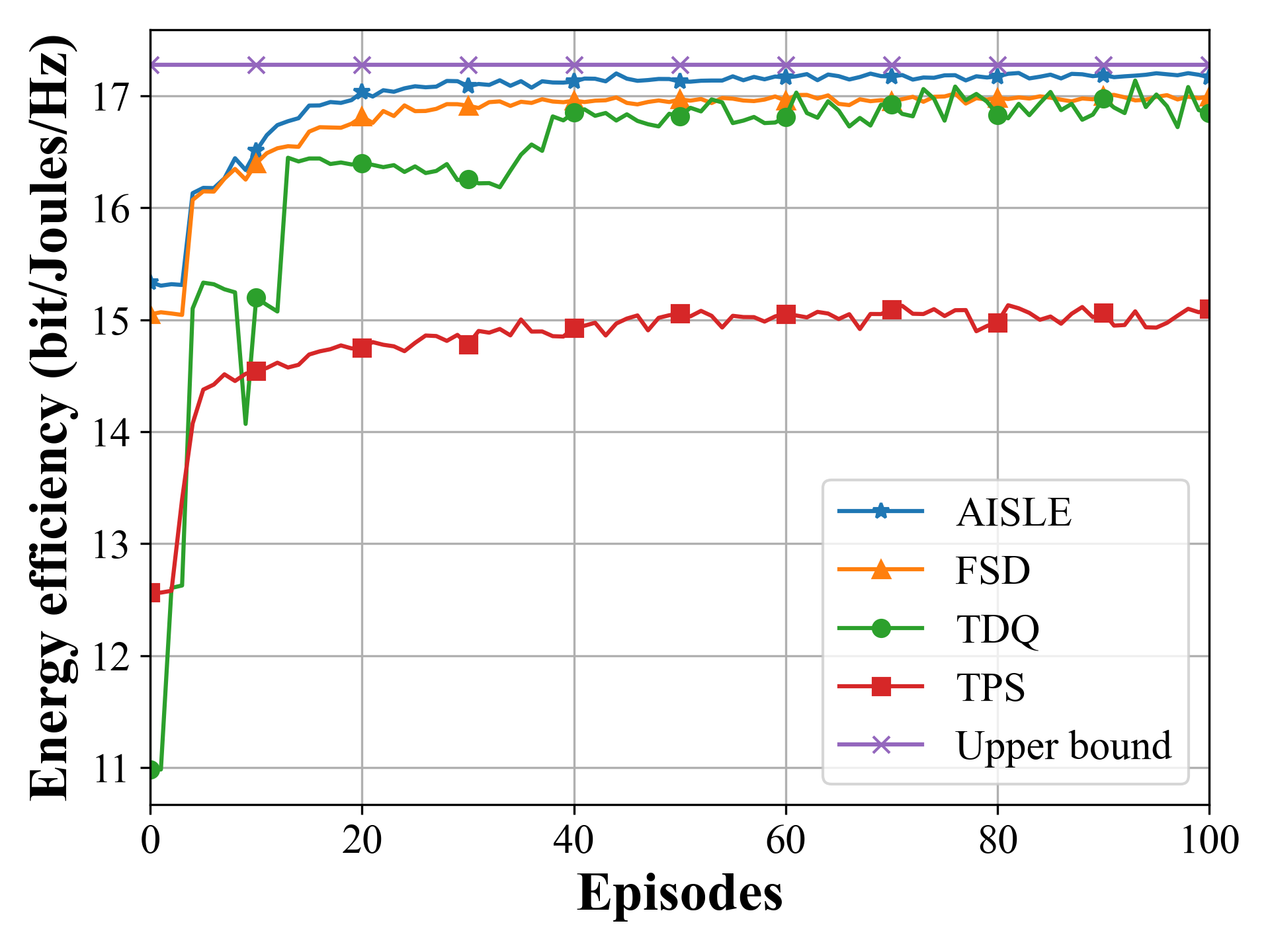}}
    \subfigure[3000]{\includegraphics[width=0.25\textwidth]{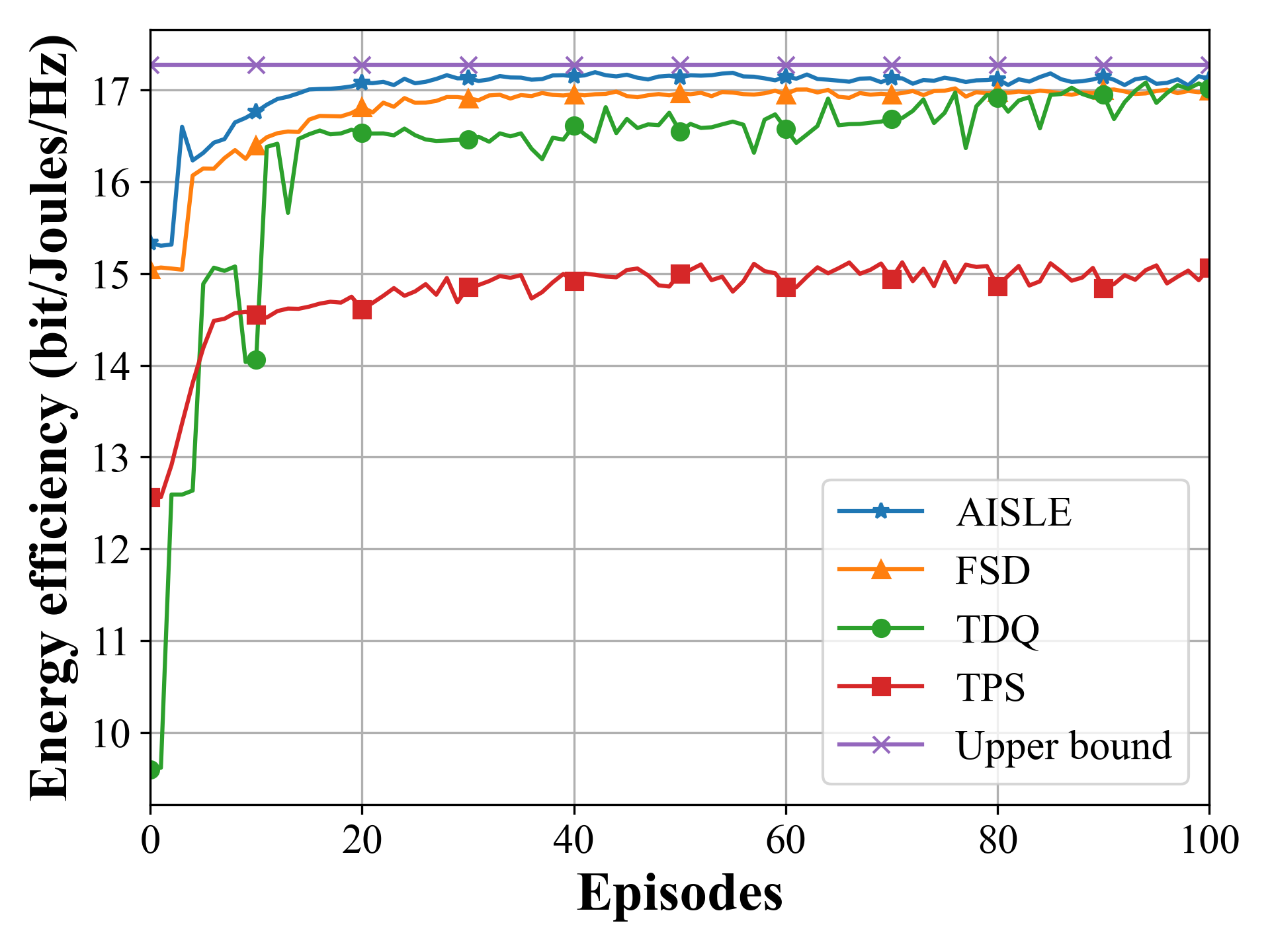}}
    \subfigure[4000]{\includegraphics[width=0.25\textwidth]{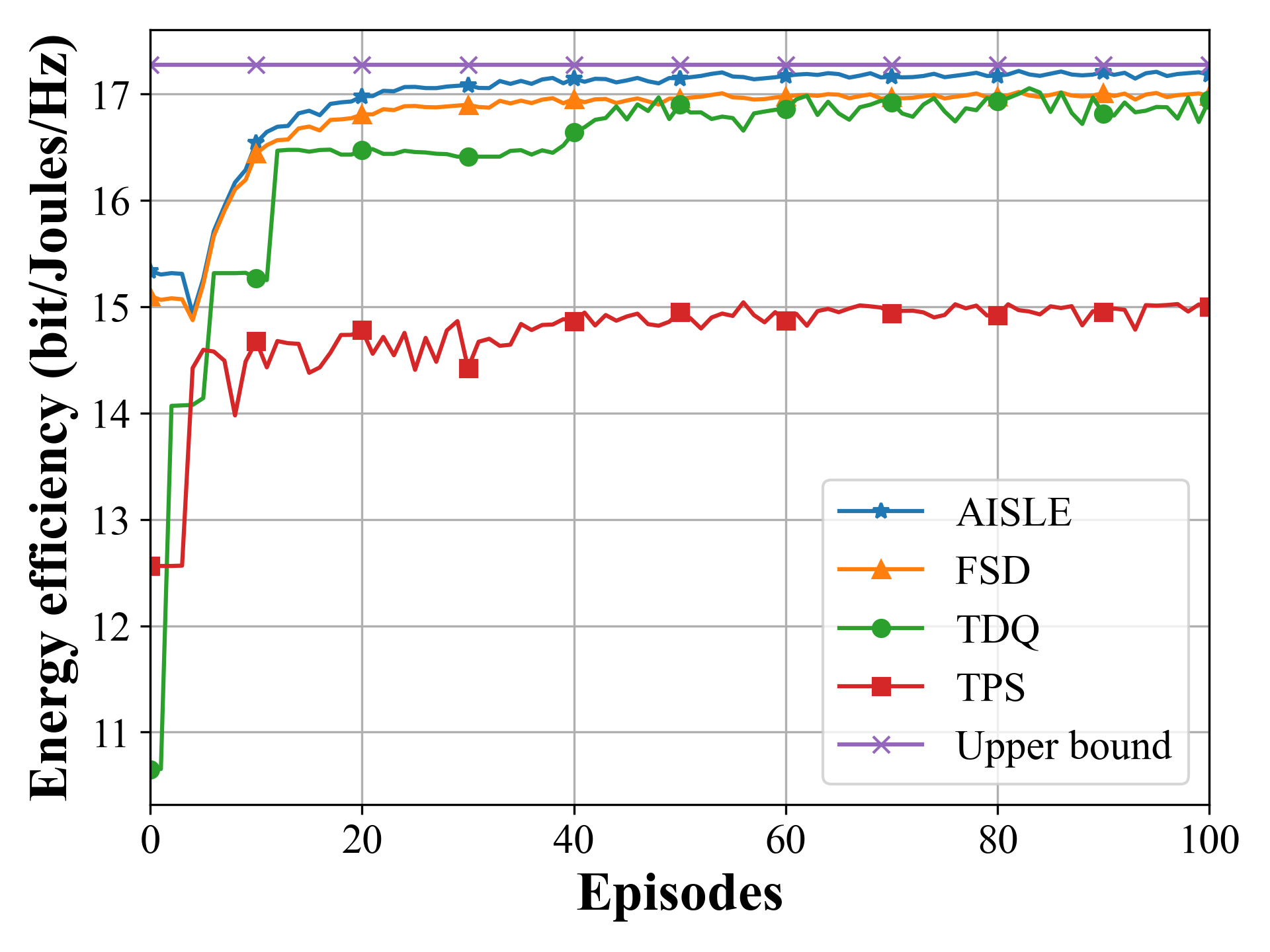}}

    \vspace{0.5cm}

   \subfigure[5000]{\includegraphics[width=0.25\textwidth]{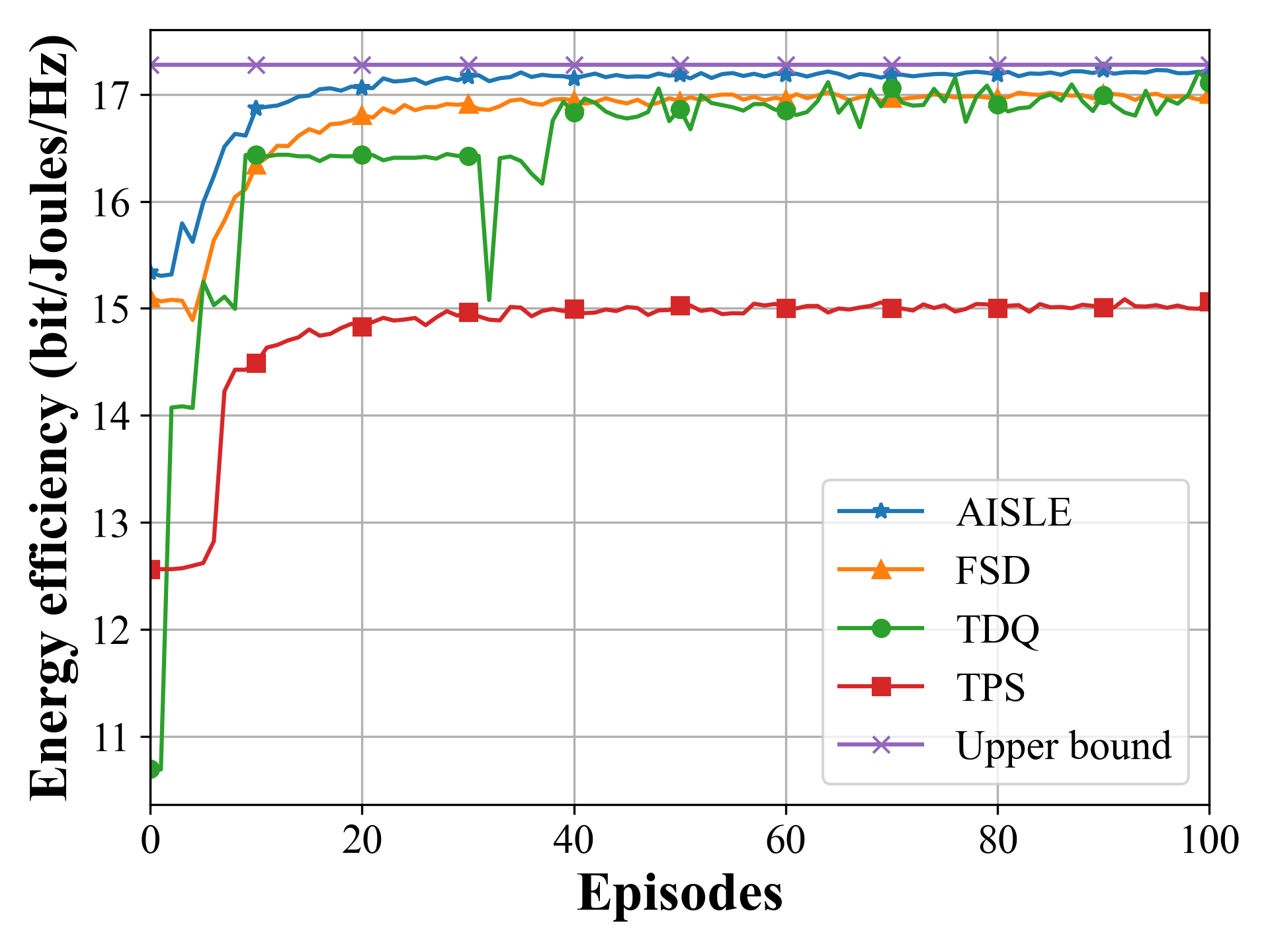}}
   \subfigure[6000]{\includegraphics[width=0.25\textwidth]{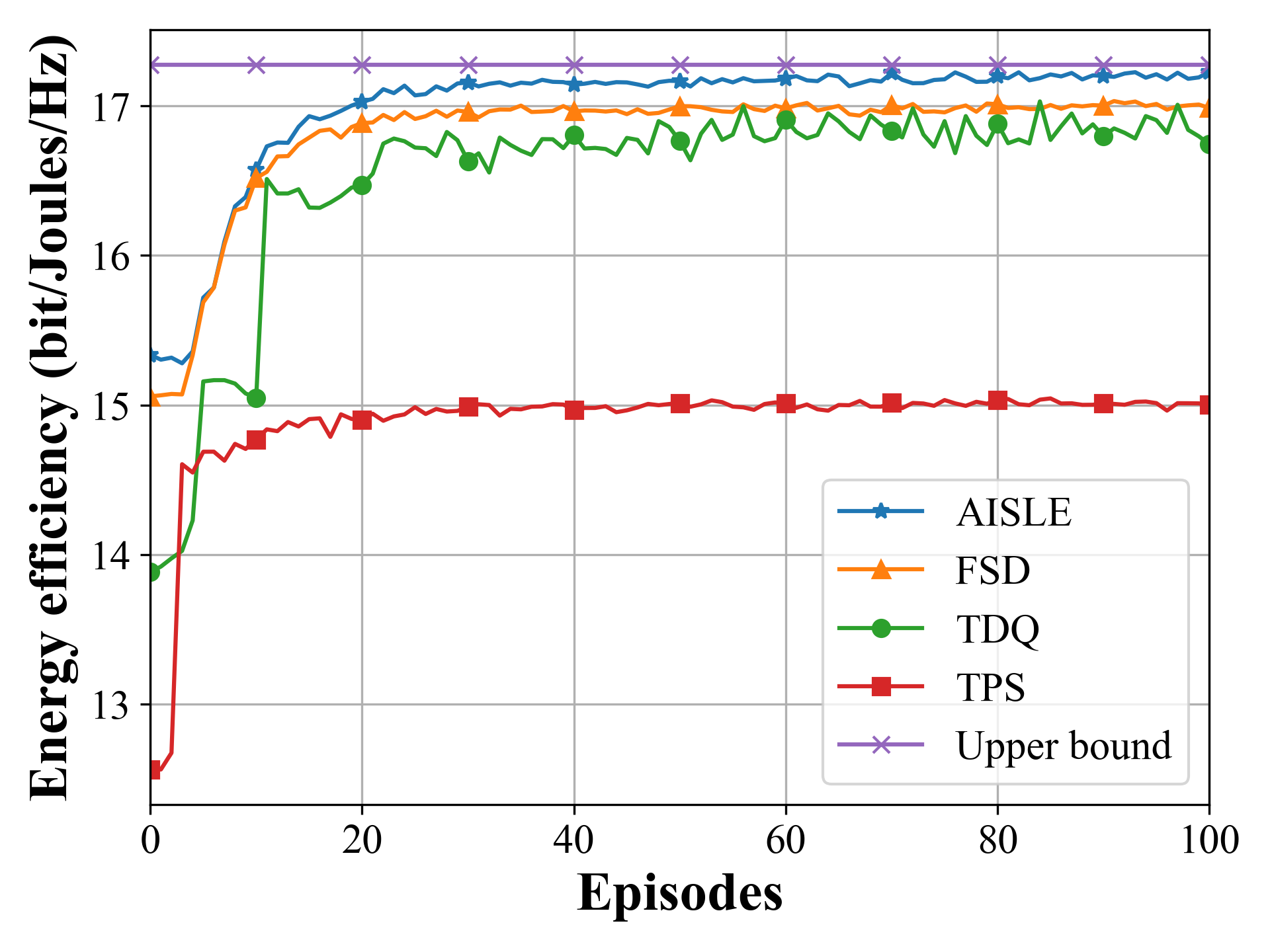}}
   \subfigure[Convergence time]{\includegraphics[width=0.25\textwidth]{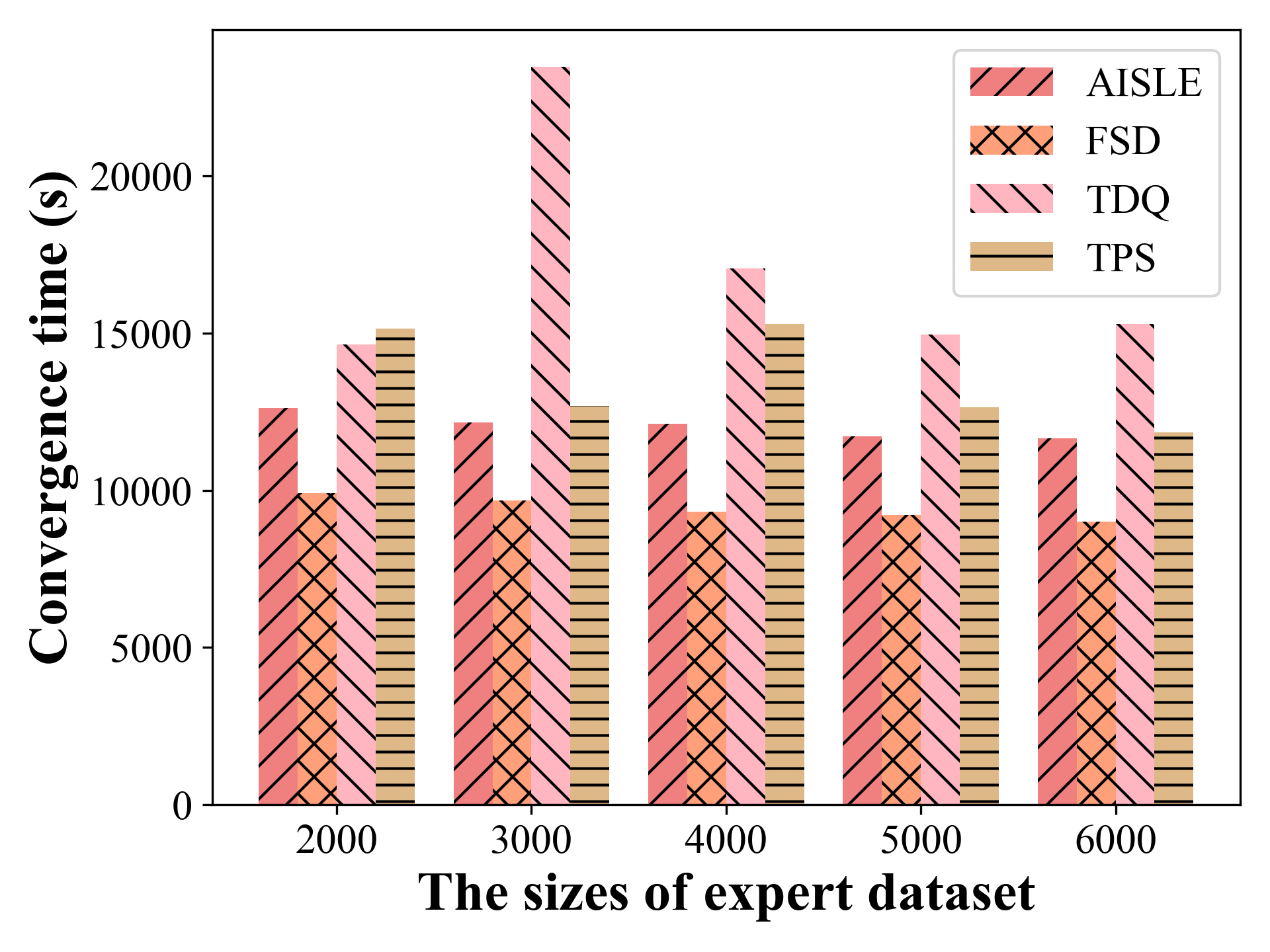}}

   \caption{Convergence performance with different expert dataset sizes.}
\end{figure*}


\indent In this section, simulation setup and performance results are presented to verify the effectiveness of AISLE.

\subsection{Simulation Setup}

\indent The simulation environment is constructed by Python 3.8, Pytorch 1.7.1, cvxpy 1.3.1, and gurobipy 10.0.1. As shown in Fig. 3, we validate the performance of AISLE on a Manhattan map. We choose the area of 500m$\times$500m indicated by the red line, where the points marked with ``X" are the deployment locations of IRSs, with randomly distributed users. And the number of users and IRSs is 3 and 2, respectively. The specific parameter settings of the simulation are referred to~\cite{9454446,9749020,9043712}. The locations of IRSs are (250 m, 250 m, 30 m) and (250 m, 0 m, 30 m), and the rest of the simulation parameters are set as follows. ${V}_{\max }=20$ m/s, $p_{max}=20$ dBm, ${J}_{b}={J}_{I}= 100$, $A=0.9$, ${\delta }^{2}=-80$ dBm, $\beta=10^{-5}$ W, $B=32$, $\eta =1$, $\gamma _{i}^{\text{UG}}=\gamma _{s,i}^{\text{RG}}=10$ dB, $\xi _{i}^{\text{UG}}=\xi _{s,i}^{\text{RG}}=2.5$. For expert datasets, we select random points different from the initial coordinate of the UAV as starting points to generate the expert dataset to avoid overfitting.

\indent Although existing related studies do not consider the combined optimization of energy consumption and transmission rate for UAV communication systems with multiple IRSs, we evaluate the proposed algorithm, AISLE, against the following three representative solutions:

\begin{itemize}
	\item TDQ~\cite{9869783}: It is a deep Q-network (DQN) and deep deterministic policy gradient (DDPG) based algorithm to optimize discrete flight trajectories and continuous flight trajectories, to maximize the energy efficiency of the system.
	
     \item TPS~\cite{10.1109/TWC.2022.3212830}: It is a double DQN (DDQN)-based algorithm to optimize trajectories and phase shifts of UAVs to maximize system transmit capacity.

	\item Fixed SIC decoding order (FSD): User association decision and trajectory optimization are the same as our algorithm, but fixed SIC decoding orders.

\end{itemize}
\subsection{Simulation Results}

\subsubsection{Performance under different expert dataset sizes}

\indent By generating expert datasets with different numbers of random UAV starting points, we can obtain different sizes of expert datasets, and an upper bound for Problem $P0$ can be obtained by generating expert datasets with UAV initial coordinates. Fig. 4 shows the performance of the AISLE algorithm with different sizes of expert datasets. We can observe that the AISLE algorithm obtains the best convergence and highest energy efficiency in all sizes of expert datasets, followed by FSD. That is because AISLE and FSD algorithms approximate the expert policy by learning directly from the expert datasets, which have faster convergence speeds compared to TDQ and TPS algorithms. The capacity of ER is constant in the AISLE algorithm, which ensures that a large amount of expert datasets can be extracted to train the model for each training. In contrast, TDQ and TPS algorithms continuously enlarge the capacity of ER when the number of episodes increases, making it difficult to utilize the expert datasets efficiently. Figs. 4(a)-3(e) show that when we gradually increase the size of the expert datasets, the performance of the AISLE algorithm is more stable than TDQ and TPS algorithms, with about 0.8\% away from the upper bound. While TDQ and TPS algorithms fluctuate with different sizes of expert datasets, we can observe that the TDQ algorithm eventually converges close to the AISLE algorithm in Figs. 4(b) and 4(d).

\indent In addition, Fig. 4(f) shows a comparison of the convergence time at different expert dataset sizes. The convergence time of the AISLE algorithm and FSD are on a slow decrease, where the AISLE algorithm decreases from 12619 seconds at an expert dataset size of 2000 to 11658 seconds at that of 6000. Note that the convergence time of FSD is shorter than that of the AISLE algorithm with different expert dataset sizes, because it simplifies Problem $P3$ by only optimizing power allocation. Thus, it requires less time to interact with the environment. We can observe that TQD and TPS exhibit random fluctuations with the increasing size of the expert dataset. This is because they train the model with a batch of state transitions obtained through random sampling across the entire ER.

\begin{figure*}[h]
	\centering
	\begin{minipage}{0.3\linewidth}
		\centering
		\includegraphics[width=0.80\linewidth]{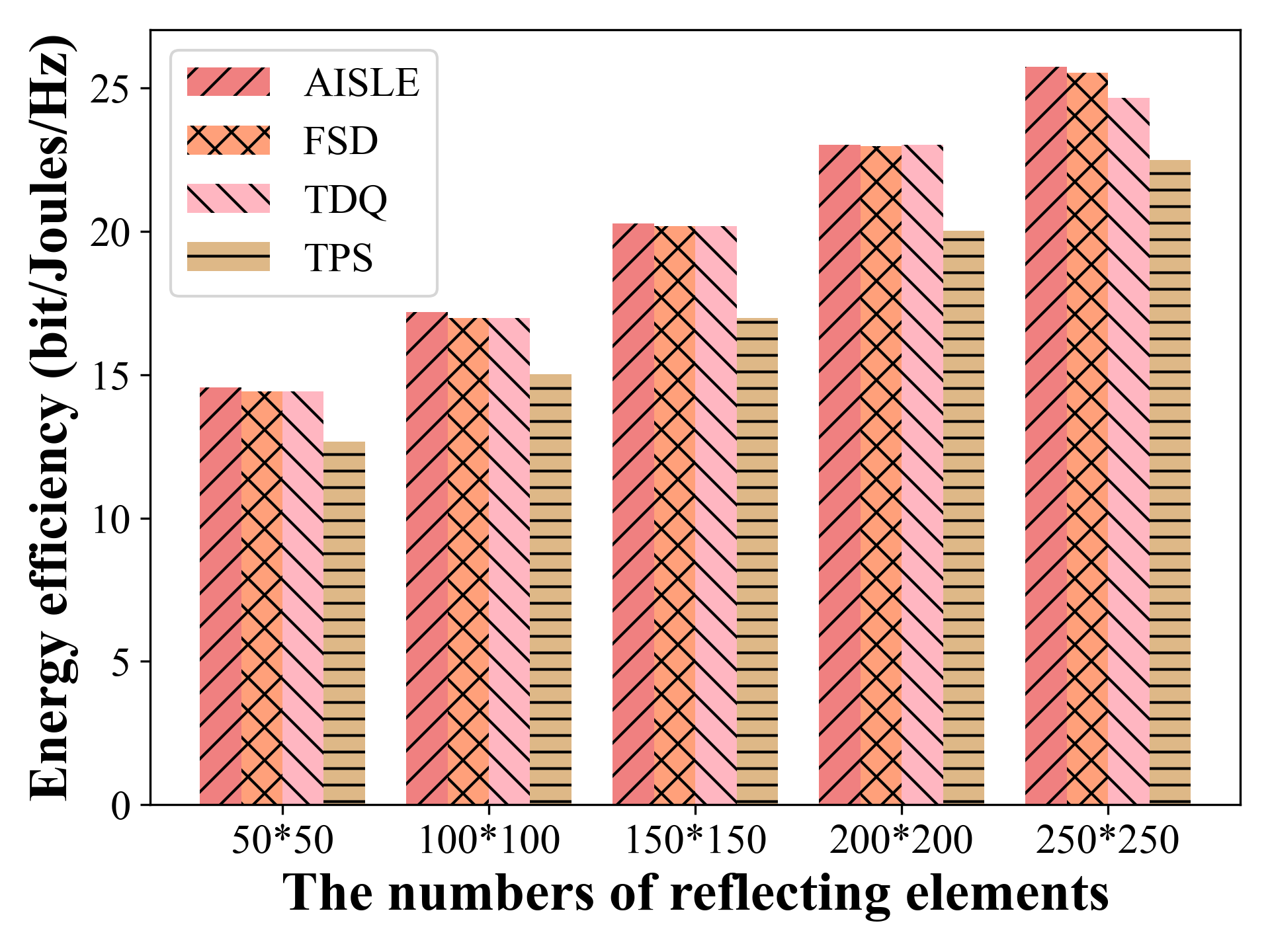}
		\caption{Performance with different numbers of reflecting elements.}
		\label{fig5}
	\end{minipage}
	\begin{minipage}{0.3\linewidth}
		\centering
		\includegraphics[width=0.80\linewidth]{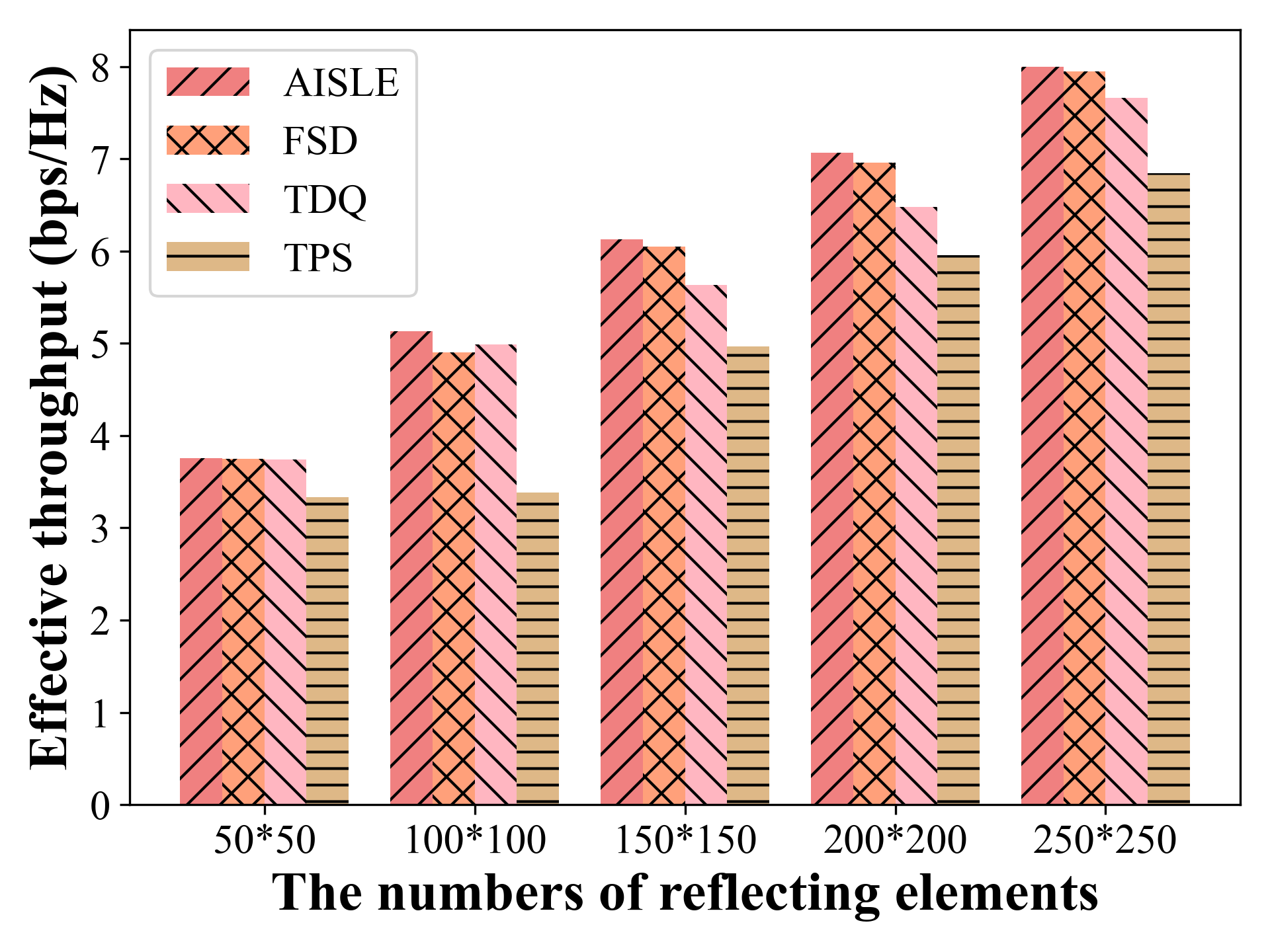}
		\caption{Performance on average effective throughput.}
		\label{fig6}
	\end{minipage}
	\begin{minipage}{0.3\linewidth}
		\centering
		\includegraphics[width=0.80\linewidth]{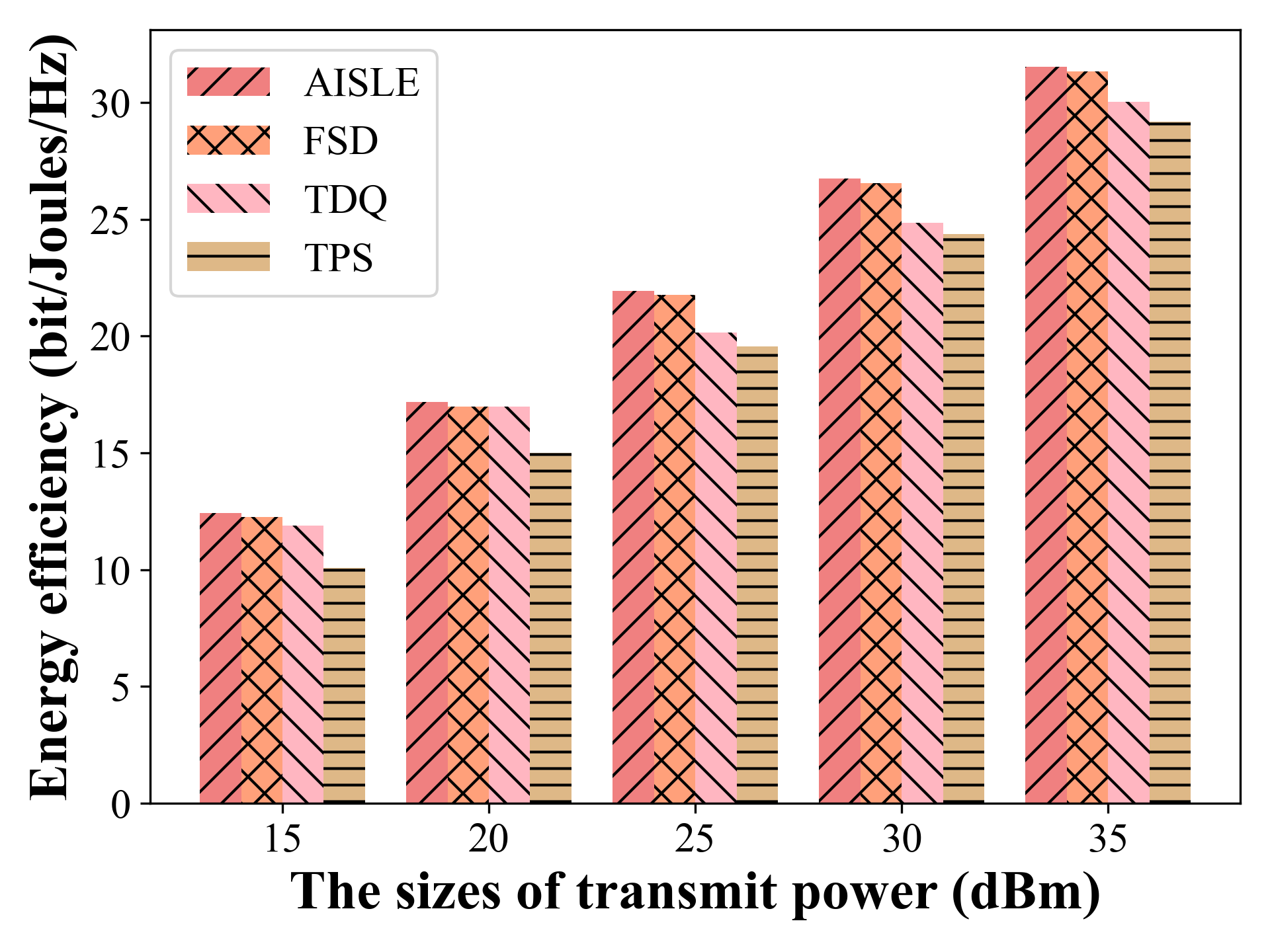}
		\caption{Performance with different transmit power.}
		\label{fig7}
	\end{minipage}

	\begin{minipage}{0.3\linewidth}
		\centering
		\includegraphics[width=0.80\linewidth]{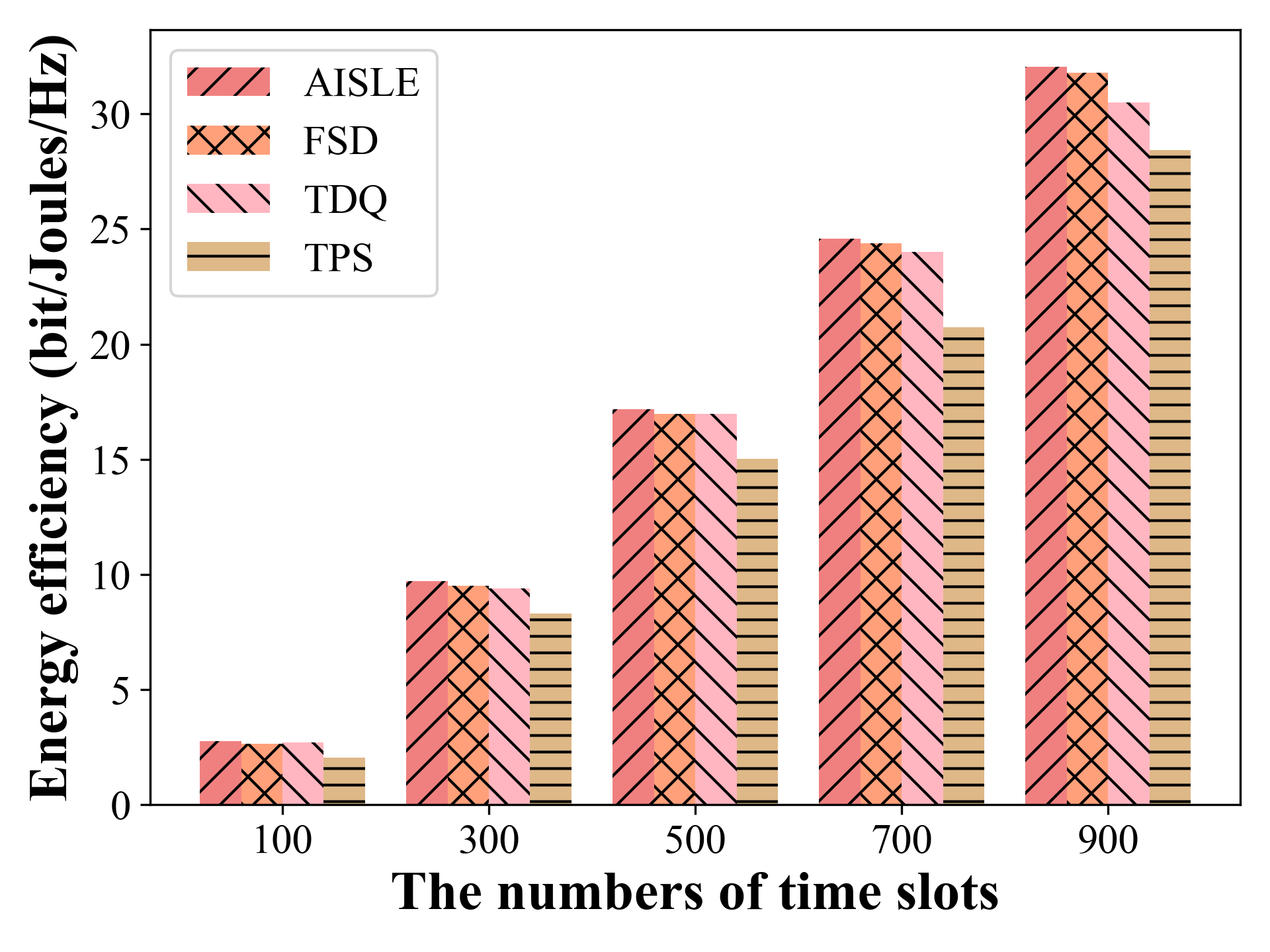}
		\caption{Performance with different numbers of time slots.}
		\label{fig8}
	\end{minipage}
	\begin{minipage}{0.3\linewidth}
		\centering
		\includegraphics[width=0.80\linewidth]{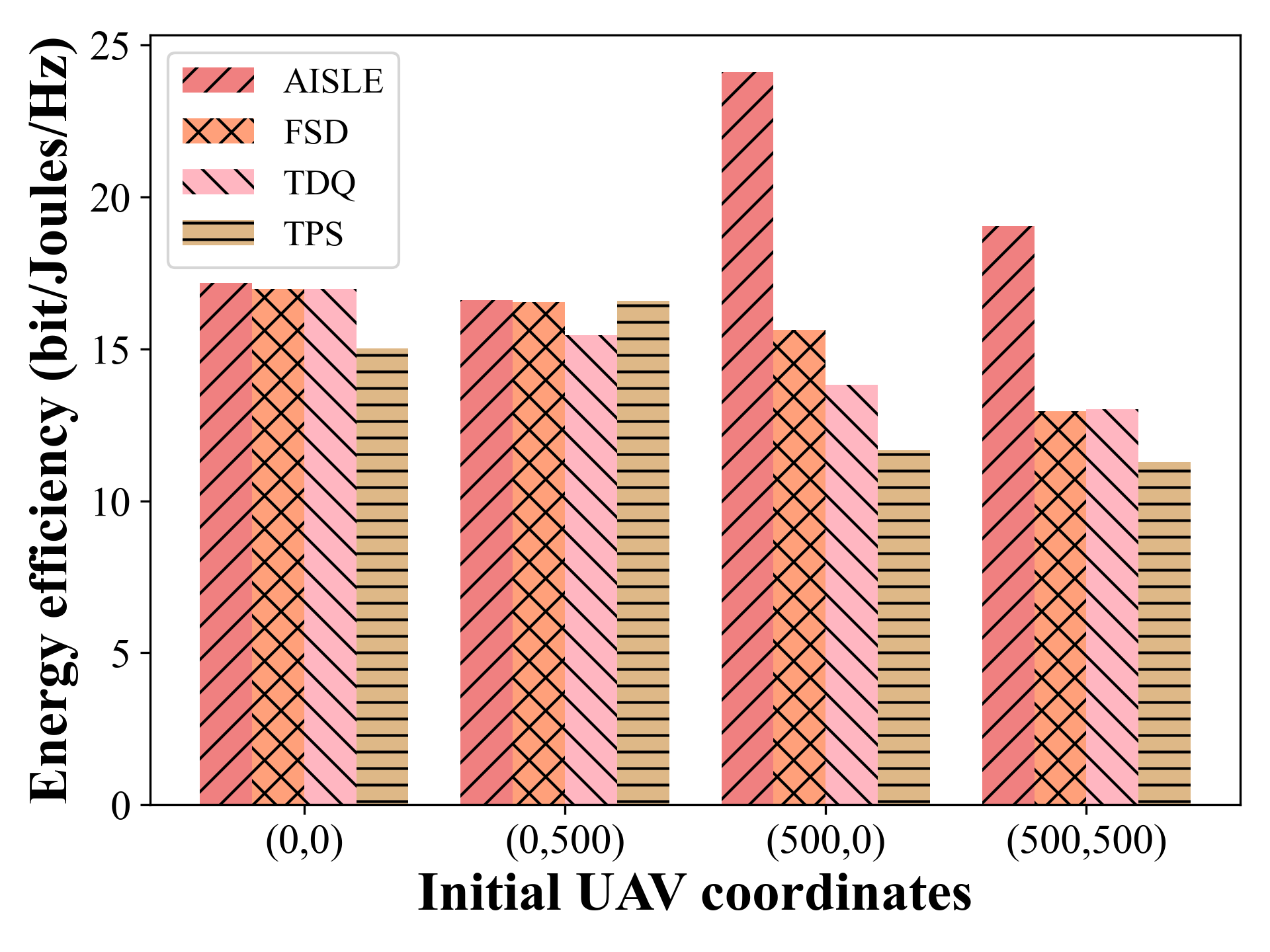}
		\caption{Performance with different initial UAV coordinates.}
		\label{fig9}
	\end{minipage}
\end{figure*}

\subsubsection{Performance under different numbers of reflecting elements}

\indent The performance on total system energy efficiency with different numbers of reflecting elements is shown in Fig. 5. It can be observed that TPS has the lowest energy efficiency, and the performance of FSD and TDQ is close, while AISLE obtains the highest energy efficiency. It is shown that the total system efficiency rises with the increasing number of reflecting elements. This arises from the channel gain expressed by equation (\ref{eq12}) is positively correlated with the number of reflecting elements. In addition, the second and third terms of equation (\ref{eq12}) increase with the number of reflecting elements. Hence, more reflective elements result in greater reflective array gain. Moreover, we can observe that the performance of AISLE is stable among different numbers of reflecting elements, while TDQ shows a degradation with 250$\times$250 reflecting elements. This is because the learning model of the AISLE algorithm is trained based on expert datasets and has good stability. Although the SIC decoding order is fixed in FSD, SCA is still applied to optimize power allocation. 

\indent Considering a minimum transmission rate threshold, we set an indicator function $\mathcal {I}\left( \cdot \right) $, i.e., $\mathcal{I}(R_{i} \ge R_{min})=1$ when $R_{i} \ge R_{min}$ and vice versa. Without loss of generality, we set $R_{min}=1$. From this, we define the average effective throughput $\tilde{R}=\frac{1}{T}\sum\limits_{t=1}^{T}{{{R}_{i}}[t]}\mathcal{I}({{R}_{i}}[t]\ge{{R}_{\min }}[t])$ and verify the effect with different numbers of IRS elements. From Fig. 6, we can observe that the performance of FSD and TDQ are close to AISLE with 50$\times$50 reflecting elements. This is due to the lower passive beamforming channel gain, resulting in a lower transmission rate and a higher probability of being less than the threshold. While the number of reflective elements is increasing, it can be observed that the AISLE algorithm achieves a significant performance in the average effective throughput. At the same time, the increase of the average effective throughput gradually flattens out, and the growth decreases from 1.4 bps/Hz at the beginning to 0.9 bps/Hz at 250$\times$250 reflecting elements. This is because the increasing number of IRS reflecting elements raises the probability that the transmission rate is higher than the threshold, and the increasing value of this probability decreases with the increasing number of IRS reflecting elements. Hence, it can be found that the average effective throughput can be effectively improved by increasing the number of IRS reflecting elements.

\subsubsection{Performance under different transmit power}

\indent Fig. 7 shows the performance of algorithms with different UAV transmit power. Similar to that shown in Fig. 6, the AISLE algorithm has the highest energy efficiency over different transmit power, followed by FSD and TDQ, and TPS has the lowest energy efficiency. It can be found that the performance of FSD is close to AISLE for different transmit power. When the transmit power increases, the energy efficiency of AISLE increases as well. This indicates that the obtained transmission rate is more substantial compared to the increase in total energy consumption. By expanding the number of reflecting elements from 100$\times$100 by a factor of 6.25 to 250$\times$250, the system energy efficiency increases by about 49.9\%, while expanding the transmit power by a factor of 10 from 20 dBm to 30 dBm increases the system energy efficiency only by 55.6\%. From the above analysis, it is clear that the gain in energy efficiency obtained by increasing the number of reflecting elements is significant. Considering the easy deployment and low-cost properties of IRSs, the system cost can be reduced to some extent by increasing the number of reflecting elements compared to the costly way of increasing transmit power.

\subsubsection{Performance under different numbers of time slots}

\indent As shown in Fig. 8, the AISLE algorithm obtains the highest system energy efficiency with different numbers of time slots. Similar to Figs. 5, 6 and 7, the performance of FSD and TDQ is close to that of the AISLE algorithm, and the performance of TPS lags. It can be noticed that the performance gap between AISLE and TDQ algorithms widens with the increase in service time, from 0.08 bit/Joules/Hz to 1.55 bit/Joules/Hz. Moreover, it can be seen that the increment of the system energy efficiency of the AISLE algorithm is about 7 bits/Joules/Hz for each additional 200 time slot. Therefore, it can be verified that the performance of the AISLE algorithm remains stable over different time slots. This is because AISLE learns from expert datasets and is capable of long-term optimization.

\subsubsection{Performance under different initial UAV coordinates}

\indent By changing the initial coordinates of the UAV, we verify the energy efficiency performance, as shown in Fig. 9. It can be seen that the performance of different algorithms changes with different initial UAV coordinates. The performance of the AISLE algorithm remains stable and achieves the highest energy efficiency at different initial UAV coordinates. The performance of FSD is close to that of AISLE when the UAV is initially positioned at coordinates $(0, 0)$ and $(0, 500)$, while the performance gap with the AISLE algorithm widens at coordinates $(500, 0)$ and $(500, 500)$. The reason is that different initial coordinates of UAVs result in varying channel conditions, and the SIC decoding orders of the FSD algorithm is fixed, which leads to poor performance stability. In addition, we can find that TPS has better performance than TDQ on $(0, 500)$ and lags on other UAV initial coordinates. It implies that the variation of the environment has a greater impact on their performance and they are overfitted on coordinate $(0, 0)$. The training process of the AISLE algorithm guarantees that half of the batches consist of expert datasets, which means the AISLE algorithm can make decisions close to the expert strategy by learning from expert data, to improve its generalization ability.

\section{Conclusion}
\indent This paper considered a multi-IRS assisted UAV communication system, in which the UAV performs as a temporary BS to serve ground users. In this scenario, we formulated an optimization problem to maximize the system energy efficiency by jointly optimizing user association, UAV trajectory, power allocation, and SIC decoding orders. First, we decomposed the optimization problem into three subproblems, i.e., user association optimization problem, UAV trajectory optimization problem, and joint power allocation and SIC decoding order optimization problem. Then, we designed an alternating optimization algorithm based on inverse soft-Q learning and incorporated Dinkelbach and penalty methods. Finally, the theoretical analysis and performance results showed that AISLE has obvious advantages in system energy efficiency and average effective throughput.

\bibliographystyle{ieeetr}
\bibliography{huhaob}

\begin{thebibliography}{10}

\bibitem{9454446}
X.~Mu, Y.~Liu, L.~Guo, J.~Lin, and H.~V. Poor, ``{Intelligent Reflecting
  Surface Enhanced Multi-UAV NOMA Networks},'' {\em IEEE J. Sel. Areas
  Commun.}, vol.~39, no.~10, pp.~3051--3066, 2021.

\bibitem{10070838}
W.~Wei, X.~Pang, J.~Tang, N.~Zhao, X.~Wang, and A.~Nallanathan, ``{Secure
  Transmission Design for Aerial IRS Assisted Wireless Networks},'' {\em IEEE
  Trans. Commun.}, vol.~71, no.~6, pp.~3528--3540, 2023.

\bibitem{10044705}
S.~Han, J.~Wang, L.~Xiao, and C.~Li, ``{Broadcast Secrecy Rate Maximization in
  UAV-Empowered IRS Backscatter Communications},'' {\em IEEE Trans. Wireless.
  Commun., doi={10.1109/TWC.2023.3243270}}, 2023.

\bibitem{9893192}
Y.~Li, H.~Zhang, K.~Long, and A.~Nallanathan, ``{Exploring Sum Rate
  Maximization in UAV-Based Multi-IRS Networks: IRS Association, UAV Altitude,
  and Phase Shift Design},'' {\em IEEE Trans. Commun.}, vol.~70, no.~11,
  pp.~7764--7774, 2022.

\bibitem{9870557}
N.~Babu, M.~Virgili, M.~Al-jarrah, X.~Jing, E.~Alsusa, P.~Popovski, A.~Forsyth,
  C.~Masouros, and C.~B. Papadias, ``{Energy-Efficient Trajectory Design of a
  Multi-IRS Assisted Portable Access Point},'' {\em IEEE Trans. Veh. Technol.},
  vol.~72, no.~1, pp.~611--622, 2023.

\bibitem{9804341}
M.~Asim, M.~ELAffendi, and A.~A.~A. El-Latif, ``{Multi-IRS and
  Multi-UAV-Assisted MEC System for 5G/6G Networks: Efficient Joint Trajectory
  Optimization and Passive Beamforming Framework},'' {\em IEEE Trans. Intell.
  Transp. Syst.}, vol.~24, no.~4, pp.~4553--4564, 2023.

\bibitem{9293155}
Z.~Wei, Y.~Cai, Z.~Sun, D.~W.~K. Ng, J.~Yuan, M.~Zhou, and L.~Sun, ``{Sum-Rate
  Maximization for IRS-Assisted UAV OFDMA Communication Systems},'' {\em IEEE
  Trans. Wireless. Commun.}, vol.~20, no.~4, pp.~2530--2550, 2021.

\bibitem{9656117}
X.~Pang, N.~Zhao, J.~Tang, C.~Wu, D.~Niyato, and K.-K. Wong, ``{IRS-Assisted
  Secure UAV Transmission via Joint Trajectory and Beamforming Design},'' {\em
  IEEE Trans. Commun.}, vol.~70, no.~2, pp.~1140--1152, 2022.

\bibitem{9894720}
S.~Zargari, A.~Hakimi, C.~Tellambura, and S.~Herath, ``{User Scheduling and
  Trajectory Optimization for Energy-Efficient IRS-UAV Networks With SWIPT},''
  {\em IEEE Trans. Veh. Technol.}, vol.~72, no.~2, pp.~1815--1830, 2023.

\bibitem{9866052}
Y.~Su, X.~Pang, S.~Chen, X.~Jiang, N.~Zhao, and F.~R. Yu, ``{Spectrum and
  Energy Efficiency Optimization in IRS-Assisted UAV Networks},'' {\em IEEE
  Trans. Commun.}, vol.~70, no.~10, pp.~6489--6502, 2022.

\bibitem{10.1109/TWC.2022.3212830}
H.~Zhang, M.~Huang, H.~Zhou, X.~Wang, N.~Wang, and K.~Long, ``{Capacity
  Maximization in RIS-UAV Networks: A DDQN-Based Trajectory and Phase Shift
  Optimization Approach},'' {\em IEEE Trans. Wireless. Commun.}, vol.~22,
  no.~4, pp.~2583--2591, 2023.

\bibitem{9749020}
X.~Zhang, H.~Zhang, W.~Du, K.~Long, and A.~Nallanathan, ``{IRS Empowered UAV
  Wireless Communication With Resource Allocation, Reflecting Design and
  Trajectory Optimization},'' {\em IEEE Trans. Wireless. Commun.}, vol.~21,
  no.~10, pp.~7867--7880, 2022.

\bibitem{9804220}
Y.~Cai, Z.~Wei, S.~Hu, C.~Liu, D.~W.~K. Ng, and J.~Yuan, ``{Resource Allocation
  and 3D Trajectory Design for Power-Efficient IRS-Assisted UAV-NOMA
  Communications},'' {\em IEEE Trans. Wireless. Commun.}, vol.~21, no.~12,
  pp.~10315--10334, 2022.

\bibitem{garg2021iqlearn}
D.~Garg, S.~Chakraborty, C.~Cundy, J.~Song, and S.~Ermon, ``{IQ-Learn: Inverse
  soft-Q Learning for Imitation},'' in {\em Thirty-Fifth Conference on Neural
  Information Processing Systems}, pp.~4028--4039, 2021.

\bibitem{9976948}
A.~Chauhan, S.~Ghosh, and A.~Jaiswal, ``{RIS Partition-Assisted Non-Orthogonal
  Multiple Access (NOMA) and Quadrature-NOMA With Imperfect SIC},'' {\em IEEE
  Trans. Wireless. Commun.}, vol.~22, no.~7, pp.~4371--4386, 2023.

\bibitem{9714216}
H.-B. Jeon, S.-H. Park, J.~Park, K.~Huang, and C.-B. Chae, ``{An
  Energy-Efficient Aerial Backhaul System With Reconfigurable Intelligent
  Surface},'' {\em IEEE Trans. Wireless. Commun.}, vol.~21, no.~8,
  pp.~6478--6494, 2022.

\bibitem{9784946}
Y.~Huang, W.~Mei, and R.~Zhang, ``{Empowering Base Stations With Co-Site
  Intelligent Reflecting Surfaces: User Association, Channel Estimation and
  Reflection Optimization},'' {\em IEEE Trans. Commun.}, vol.~70, no.~7,
  pp.~4940--4955, 2022.

\bibitem{9810528}
Z.~Ji, W.~Yang, X.~Guan, X.~Zhao, G.~Li, and Q.~Wu, ``{Trajectory and Transmit
  Power Optimization for IRS-Assisted UAV Communication Under Malicious
  Jamming},'' {\em IEEE Trans. Veh. Technol.}, vol.~71, no.~10,
  pp.~11262--11266, 2022.

\bibitem{10109153}
K.~Guo, M.~Wu, X.~Li, H.~Song, and N.~Kumar, ``{Deep Reinforcement Learning and
  NOMA-Based Multi-Objective RIS-Assisted IS-UAV-TNs: Trajectory Optimization
  and Beamforming Design},'' {\em IEEE Trans. Intell. Transp. Syst.}, vol.~24,
  no.~9, pp.~10197--10210, 2023.

\bibitem{9913496}
S.~K. Singh, K.~Agrawal, K.~Singh, C.-P. Li, and Z.~Ding, ``{NOMA Enhanced
  Hybrid RIS-UAV-Assisted Full-Duplex Communication System With Imperfect SIC
  and CSI},'' {\em IEEE Trans. Commun.}, vol.~70, no.~11, pp.~7609--7627, 2022.

\bibitem{9681874}
J.~Zhao, L.~Yu, K.~Cai, Y.~Zhu, and Z.~Han, ``{RIS-Aided Ground-Aerial NOMA
  Communications: A Distributionally Robust DRL Approach},'' {\em IEEE J. Sel.
  Areas Commun.}, vol.~40, no.~4, pp.~1287--1301, 2022.

\bibitem{9043712}
Y.~Cai, Z.~Wei, R.~Li, D.~W.~K. Ng, and J.~Yuan, ``{Joint Trajectory and
  Resource Allocation Design for Energy-Efficient Secure UAV Communication
  Systems},'' {\em IEEE Trans. Commun.}, vol.~68, no.~7, pp.~4536--4553, 2020.

\bibitem{9849460}
J.~Zhao, Y.~Zhu, X.~Mu, K.~Cai, Y.~Liu, and L.~Hanzo, ``{Simultaneously
  Transmitting and Reflecting Reconfigurable Intelligent Surface (STAR-RIS)
  Assisted UAV Communications},'' {\em IEEE J. Sel. Areas Commun.}, vol.~40,
  no.~10, pp.~3041--3056, 2022.

\bibitem{9769985}
K.~K. Nguyen, A.~Masaracchia, V.~Sharma, H.~V. Poor, and T.~Q. Duong,
  ``{RIS-Assisted UAV Communications for IoT With Wireless Power Transfer Using
  Deep Reinforcement Learning},'' {\em IEEE J. Sel. Topics Signal Process.},
  vol.~16, no.~5, pp.~1086--1096, 2022.

\bibitem{9277627}
X.~Liu, Y.~Liu, and Y.~Chen, ``{Machine Learning Empowered Trajectory and
  Passive Beamforming Design in UAV-RIS Wireless Networks},'' {\em IEEE J. Sel.
  Areas Commun.}, vol.~39, no.~7, pp.~2042--2055, 2021.

\bibitem{9953122}
S.~Guo and X.~Zhao, ``{Multi-Agent Deep Reinforcement Learning Based
  Transmission Latency Minimization for Delay-Sensitive Cognitive Satellite-UAV
  Networks},'' {\em IEEE Trans. Commun.}, vol.~71, no.~1, pp.~131--144, 2023.

\bibitem{9494520}
T.~Wang, Y.~Li, and Y.~Wu, ``{Energy-Efficient UAV Assisted Secure Relay
  Transmission via Cooperative Computation Offloading},'' {\em IEEE Trans.
  Green Commun. Netw.}, vol.~5, no.~4, pp.~1669--1683, 2021.

\bibitem{9839554}
X.~Fang, W.~Feng, Y.~Wang, Y.~Chen, N.~Ge, Z.~Ding, and H.~Zhu, ``{NOMA-Based
  Hybrid Satellite-UAV-Terrestrial Networks for 6G Maritime Coverage},'' {\em
  IEEE Trans. Wireless. Commun.}, vol.~22, no.~1, pp.~138--152, 2023.

\bibitem{9869783}
L.~Wang, K.~Wang, C.~Pan, and N.~Aslam, ``{Joint Trajectory and Passive
  Beamforming Design for Intelligent Reflecting Surface-Aided UAV
  Communications: A Deep Reinforcement Learning Approach},'' {\em IEEE Trans.
  Mobile Comput.}, vol.~22, no.~11, pp.~6543--6553, 2023.

\end{thebibliography}

\end{document}